\documentclass[draftcls, 11pt, onecolumn]{IEEEtran}

\usepackage[dvips]{graphics}
\usepackage{graphicx}
\usepackage{times}
\usepackage{amsmath}
\usepackage{amssymb}
\usepackage{fancybox}
\newtheorem{property}{\bf Property}
\newtheorem{lemma}{\bf Lemma}
\newtheorem{theorem}{\bf Theorem}

\newcommand{\bee}{\begin{eqnarray}}
\newcommand{\eee}{\end{eqnarray}}
\newcommand{\be}{\begin{equation}}
\newcommand{\ee}{\end{equation}}

\newcommand{\al}[1]{\begin{align} #1 \end{align}}

\newcommand{\equ}[1]{\begin{equation} #1 \end{equation}}

\newcommand{\mb}{\mathbf}

\newcommand{\bs}{\boldsymbol}
\newcommand{\wt}{\widetilde}
\newcommand{\nnb}{\nonumber}

\newcommand{\qa}{{\bf a}}

\newcommand{\qe}{{\bf e}}

\newcommand{\qh}{{\bf h}}
\newcommand{\qn}{{\bf n}}

\newcommand{\qr}{{\bf r}}
\newcommand{\qs}{{\bf s}}

\newcommand{\qu}{{\bf u}}
\newcommand{\qv}{{\bf v}}

\newcommand{\qx}{{\bf x}}
\newcommand{\qy}{{\bf y}}

\newcommand{\qA}{{\bf A}}
\newcommand{\qB}{{\bf B}}
\newcommand{\qC}{{\bf C}}
\newcommand{\qD}{{\bf D}}

\newcommand{\qF}{{\bf F}}

\newcommand{\qH}{{\bf H}}
\newcommand{\qI}{{\bf I}}

\newcommand{\qK}{{\bf K}}

\newcommand{\qP}{{\bf P}}
\newcommand{\qQ}{{\bf Q}}
\newcommand{\qR}{{\bf R}}
\newcommand{\qS}{{\bf S}}

\newcommand{\qU}{{\bf U}}

\newcommand{\qW}{{\bf W}}
\newcommand{\qX}{{\bf X}}

\begin{document}
%
\title{Transmitter Optimization for Achieving Secrecy Capacity in Gaussian MIMO Wiretap Channels}

\author{\IEEEauthorblockN{Jiangyuan Li and Athina Petropulu }\\
\IEEEauthorblockA{Department of Electrical and Computer Engineering\\
Drexel University, Philadelphia, PA 19104}\\
\IEEEauthorblockA{Email: eejyli@yahoo.com.cn, athina@coe.drexel.edu}
}

\maketitle

\begin{abstract}
We consider a  Gaussian multiple-input multiple-output (MIMO) wiretap channel model,
where there exists a transmitter, a legitimate receiver and an eavesdropper, each node equipped with multiple antennas.
We study the problem of finding the optimal input covariance matrix that achieves secrecy capacity subject to a power constraint,
which leads to a non-convex optimization problem that is in general difficult to solve. Existing results for this problem address the case in
which the transmitter and the legitimate receiver have two antennas each and the eavesdropper has one antenna.
For the general cases, it has been shown that
the optimal input covariance matrix has low rank when the difference between the Grams of
the eavesdropper and the legitimate receiver channel matrices is indefinite or semi-definite,
while it may have low rank or full rank when the difference is positive definite.
In this paper, the aforementioned non-convex optimization problem is investigated.
In particular, for the multiple-input single-output (MISO) wiretap channel, the optimal input covariance matrix is obtained in closed form.
For general cases, we derive the necessary conditions for the optimal input covariance matrix consisting of a set of equations.
For the case in which the transmitter has two antennas, the derived necessary conditions can result in
a closed form solution;
For the case in which the difference between the Grams is indefinite and has all negative eigenvalues except one positive eigenvalue,
the optimal input covariance matrix has rank one and can be obtained in closed form;
For other cases,
the solution is proved to be a fixed point of a mapping from a convex set to itself and an iterative procedure is provided to
search for it.
Numerical results are presented to illustrate the proposed
theoretical findings.
\end{abstract}

\IEEEpeerreviewmaketitle

\section{Introduction}

Wireless physical (PHY) layer based security from a information-theoretic point
of view has received considerable attention recently, e.g., \cite{Liang3}-\cite{Liu},
and the comprehensive overview in \cite{Liang2}.
Wireless PHY layer based security approaches exploit the physical characteristics of the
wireless channel to enhance the security of communication systems.
The wiretap channel, first introduced and studied by Wyner \cite{Wyner}, is the most basic physical layer model that captures
the problem of communication security.
Wyner showed that when an eavesdropper's channel is a degraded version of the main channel,
the source and destination can achieve a positive perfect information rate
(secrecy rate).
The maximal rate of secrecy rate from the source to the destination
is defined as the {\it secrecy capacity} and for the degraded wiretap channel is given as the
 largest between  zero and
the difference between the capacity at the legitimate receiver and the capacity at the eavesdropper.
The Gaussian wiretap channel, in which the outputs at the legitimate
receiver and at the eavesdropper are corrupted by additive white
Gaussian noise (AWGN), was studied in \cite{Hellman}.
Along the the same line, the Gaussian MIMO wiretap channel was investigated
and the secrecy capacity of the MIMO wiretap channel was established
in terms of an optimization problem over all possible input covariance matrices \cite{Khisti}, \cite{Hassibi}.
In \cite{Khisti}, \cite{Hassibi}, the Gaussian MIMO wiretap channel model
was given as $\qy_i=\qH_i\qx+\qn_i$, $i=1,2$ where $\qn_i$ is AWGN with zero mean and covariance $\sigma^2\qI$,
and the power constraint $\mathrm{Tr}(\qR_x)\le P$ was used, where $\qR_x$ is the input covariance matrix.
An alternative expression of secrecy capacity was derived in \cite{Shamai}, \cite{Poor}
for another Gaussian MIMO wiretap channel model, i.e.,  $\qy_i=\qx+\qv_i$, $i=1,2$
where $\qv_i$ is additive Gaussian noise (AGN) with zero mean and invertible covariance $\qW_i$.
Further, in \cite{Poor}, the power covariance constraint $\qR_x\preceq \qS$ was used where $\qS$ is a given matrix and
$\qR_x\preceq \qS$ denotes that $\qS-\qR_x$ is positive semi-definite, which allowed
for the secrecy capacity to be obtained in
closed form rather than as a solution to an optimization problem.

For the former Gaussian MIMO wiretap channel model, i.e.,  $\qy_i=\qH_i\qx+\qn_i$ and the power constraint $\mathrm{Tr}(\qR_x)\le P$,
finding the optimal input covariance matrix that achieves secrecy capacity
leads to a non-convex optimization problem. This problem is in general difficult to solve.
The solution of a special case in which the transmitter and the legitimate receiver each has two antennas and the
eavesdropper has one antenna was given in \cite{Shafiee}.
In \cite{Hassibi2}, it was pointed out that
the optimal input covariance matrix has low rank when the difference between the Grams of
the eavesdropper and the legitimate receiver channel matrices is indefinite or semi-definite,
based on the assumption that the Grams both have full rank.

In this paper, we investigate the aforementioned non-convex optimization problem.
In particular, for the multiple-input single-output (MISO) wiretap channel, we obtain the optimal input covariance matrix in closed form.
For general MIMO case, we derive the necessary conditions for the optimal solution consisting of a set of equations.
Those conditions   result in
a closed form solution for $n_T=2$. For the more general case in which the difference of the Grams of
the eavesdropper and the legitimate receiver channel matrices is indefinite and has all negative eigenvalues except one positive eigenvalue,
we prove that the optimal input covariance matrix has rank one and can be obtained in closed form.
Otherwise,
we prove that the solution is a fixed point of a mapping from a convex set to itself and provide an iterative procedure to
search for it.

The remainder of this paper is organized as follows. The
mathematical model is introduced in \S\ref{sec:2}.
In \S\ref{sec:3} the optimal input covariance matrix is obtained in closed form
for Gaussian MISO wiretap channel.
In \S\ref{sec:4} we derive the necessary conditions for the optimal solution consisting of a set of equations
for the Gaussian MIMO wiretap channel. In \S\ref{nT2},
we obtain a closed form solution for the case in which the transmitter has
two antennas. In \S\ref{OnePosiEig},
for the case in which the difference of the Grams is indefinite and has all negative eigenvalues
except one positive eigenvalue,
we prove that the optimal input covariance matrix has rank one and can be obtained in closed form.
Numerical results in \S\ref{sec:sim} illustrate the proposed algorithm. Finally,
\S\ref{sec:conclu} gives a brief conclusion.

\subsection{Notation}

Upper case and lower case bold symbols denote matrices and vectors, respectively.
Superscripts $\ast$, $T$ and $\dagger$
denote respectively conjugate, transposition and conjugate transposition. $\mathrm{det}(\qA)$ and
$\mathrm{Tr}({\mb A})$ denote the determinant and trace of matrix $\mb A$, respectively.
$\lambda_{\max}(\qA)$ denotes the largest eigenvalue of $\qA$.
${\mb A}\succeq 0$ and ${\mb A}\succ 0$ mean that ${\mb A}$ is a
Hermitian positive semi-definite and positive definite matrix, respectively.
$\qA\succeq \qB$ denotes that $\qA-\qB$ is a positive semi-definite matrix.
${\mathbf A}\nsucc 0$ denotes that matrix $\qA$ is not positive definite.
$\mathrm{rank}(\qA)$ denotes the rank of matrix $\qA$.
$\mathrm{diag}(\qv)$ denotes a diagonal matrix with diagonal entries consisting of the elements of $\qv$.
$\|\qa\|$ denotes Euclidean norm of vector $\qa$, while $\|\qA\|$ denotes Frobenius norm of matrix $\qA$.
$\qI_n$ denotes the identity matrix of order $n$
(the subscript is dropped when the dimension is obvious).
$a\to b$ means that $a$ goes to $b$.
Given a matrix $\qA$, the matrix $\qA^\dagger\qA$ is called a Gram.

\section{System Model and Formulations}\label{sec:2}

Consider a MIMO wiretap channel where the transmitter is equipped with $n_T$ antennas, while
the legitimate receiver and an eavesdropper have  $n_R$ and $n_E$ antennas, respectively.
The received signals at the legitimate receiver and the eavesdropper are respectively given by
\equ{
\qy_R=\qH_R\qx+\qn_R,\, \qy_E=\qH_E\qx+\qn_E
}
where $\qH_R$ ($n_R\times n_T$), $\qH_E$ ($n_E\times n_T$) are respectively channel matrices between the transmitter and legitimate receiver,
and between the transmitter and eavesdropper;
$\qx$ is the $n_T\times 1$ transmitted signal vector with zero mean and $n_T \times n_T$ covariance matrix $\qR_x\succeq 0$;
$\qn_R$ and $\qn_E$ are circular Gaussian noise vectors with zero mean and covariance matrices
$\sigma^2\qI_{n_R}$ and $\sigma^2\qI_{n_E}$, respectively.
We assume the power constraint is $P$, namely, $\mathrm{Tr}(\qR_x)=P$. It is easy to verify that the problem under
$\mathrm{Tr}(\qR_x)\le P$ is equivalent to that under $\mathrm{Tr}(\qR_x)=P$.
We consider the scenario in which the transmitter has perfect short-term channel state information (CSI).

The secrecy capacity is defined as \cite{Hassibi}
\equ{
C_s\triangleq\max_{\qR_x\succeq 0, \mathrm{Tr}(\qR_x)=P}\ C_s(\qR_x)\label{SecCapObj}
}
where
\begin{equation}\label{SecRate}
    C_s(\qR_x)=\log\mathrm{det}(\qI_{n_R}+\qH_R\qR_x\qH_R^\dagger/\sigma^2)
    -\log\mathrm{det}(\qI_{n_E}+\qH_E\qR_x\qH_E^\dagger/\sigma^2)
\end{equation}
is the secrecy rate.

The transmitter optimization problem is to determine $\qR_x$ that maximizes the secrecy rate, i.e.,
achieves secrecy capacity.
The optimization makes sense when the secrecy capacity is positive.
We assume $\qH_R^\dagger\qH_R-\qH_E^\dagger\qH_E\ne 0$ since otherwise $C_s(\qR_x)\equiv 0$.
Whether $C_s>0$ depends on the  difference between the Grams of the legitimate receiver and eavesdropper channel matrices, i.e.,
$\qH_R^\dagger\qH_R-\qH_E^\dagger\qH_E$.
The following lemma provides  the conditions to maintain $C_s>0$.

\begin{lemma}\label{lem:1}
{\it The sufficient and necessary condition for $C_s>0$ is: $\qH_R^\dagger\qH_R-\qH_E^\dagger\qH_E$
is positive semi-definite or indefinite.}
\end{lemma}

Please see Appendix \ref{proofLem1} for details.

Let us assume that the channel matrices $\qH_R$ and $\qH_E$ have been normalized so that $\mathrm{Tr}(\qH_R^\dagger\qH_R)=n_T$.
We denote the signal-to-noise ratio (SNR) $\rho \triangleq P/\sigma^2$ and let $\qR_x=P\qQ$. The constraints in (\ref{SecCapObj}) now become
$\qQ\succeq 0$, $\mathrm{Tr}(\qQ)=1$. The secrecy rate
maximization problem can be written as
\equ{
\max_{\qQ\succeq 0, \mathrm{Tr}(\qQ)=1} C_s(\qQ)=
\log\mathrm{det}(\qI_{n_R}+\rho\qH_R\qQ\qH_R^\dagger)-\log\mathrm{det}(\qI_{n_E}+\rho\qH_E\qQ\qH_E^\dagger).\label{SecRateMax1}
}
According to Lemma \ref{lem:1}, we assume
$\qH_R^\dagger\qH_R-\qH_E^\dagger\qH_E$ is positive semi-definite or indefinite (except the MISO channel).
In \cite{Hassibi2}, the authors assumed that $\qH_R^\dagger\qH_R$ and $\qH_E^\dagger\qH_E$ are both positive definite,
and hence they are both invertible.
Here, we do not make that assumption.
In fact, when $n_T>n_R$ (and/or $n_T>n_E$), $\qH_R^\dagger\qH_R$ (and/or $\qH_E^\dagger\qH_E$)
always have low rank and hence are not invertible.

We also denote the feasible set of (\ref{SecRateMax1}) as
\equ{
\Omega=\{\qQ|\qQ\succeq 0, \mathrm{Tr}(\qQ)=1\}
}
which is a convex set.

\section{Closed Form Secrecy Capacity of MISO Wiretap Channel}\label{sec:3}

We first provide a lemma that will be used here and in the proof of Theorem \ref{theo:2} later.

\begin{lemma}\label{lem:2}
{\it Let $\qr$ and $\qs$ be two known non-zero vectors, and $\qr\qr^\dagger-\qs\qs^\dagger\ne 0$.
\begin{enumerate}
    \item [(i)] If $\qr=\xi\qs$ for a certain scalar $\xi$, $\qr\qr^\dagger-\qs\qs^\dagger$ has only one
    nonzero eigenvalue equal to $(|\xi|^2-1)\|\qs\|^2$ with the associated eigenvector $\qs/\|\qs\|$;
    \item [(ii)] If $\qr^\dagger\qs=0$, $\qr\qr^\dagger-\qs\qs^\dagger$ has only two nonzero eigenvalues, i.e.,
    $\eta_1=\|\qr\|^2$, $\eta_2=-\|\qs\|^2$ with associated eigenvectors $\qr/\|\qr\|$, $\qs/\|\qs\|$, respectively.
    \item [(iii)] If neither $\qr=\xi\qs$ nor $\qr^\dagger\qs=0$, $\qr\qr^\dagger-\qs\qs^\dagger$ has only two nonzero
eigenvalues, i.e.,
$\eta_1=\|\qr\|^2-|c_2||\qr^\dagger\qs|>0,
\eta_2=\|\qr\|^2-|c_4||\qr^\dagger\qs|<0$
with the associated eigenvectors
$\qe_1=c_1^{-1/2}(\qr+|c_2|e^{\mathrm{i}(\pi-\varphi)}\qs),
\qe_2=c_3^{-1/2}(\qr+|c_4|e^{\mathrm{i}(\pi-\varphi)}\qs)$, respectively,
where
$\varphi$ is the argument of
$\qr^\dagger\qs$, $\mathrm{i}=\sqrt{-1}$,
$c_1=\|\qr\|^2+|c_2|^2\|\qs\|^2-2|c_2||\qr^\dagger\qs|$,
$|c_2|
=(\|\qr\|^2+\|\qs\|^2-\sqrt{(\|\qr\|^2+\|\qs\|^2)^2-4|\qr^\dagger\qs|^2})/(2|\qr^\dagger\qs|)$,
$c_3=\|\qr\|^2+|c_4|^2\|\qs\|^2-2|c_4||\qr^\dagger\qs|$,
$|c_4|=(\|\qr\|^2+\|\qs\|^2+\sqrt{(\|\qr\|^2+\|\qs\|^2)^2-4|\qr^\dagger\qs|^2})/(2 |\qr^\dagger\qs|)$.
\end{enumerate}}
\end{lemma}
The proof is simple, therefore, omitted for the sake of brevity. But we outline the proof here. For the case (i), (ii), the proof is obvious.
For the case (iii), first, we can show $\qr\qr^\dagger-\qs\qs^\dagger$ has rank two, thus it has only two nonzero eigenvalues.
Second, we assume the eigenvector has the form of a linear combination of $\qr$ and $\qs$, and then show that this is indeed the case.

Before discussing the general MIMO wiretap channel, we analyze a special case, i.e.,
the MISO wiretap channel in which
the legitimate receiver and eavesdropper both have a single antenna, i.e., $n_R=n_E=1$. Denote the channel vectors as $\qh_R$ and $\qh_E$.
We give the following theorem.
\begin{theorem}\label{theo:MISO}
{\it The closed form expression for secrecy capacity of MISO wiretap channel is given by
\equ{
C_s=\log\frac{b+\sqrt{b^2-4ac}}{2a}\label{CsMISO}
}
where $a=1+\rho\|\qh_E\|^2$, $b=2+\rho\|\qh_R\|^2+\rho\|\qh_E\|^2+\rho^2(\|\qh_R\|^2\|\qh_E\|^2-|\qh_R^\dagger\qh_E|^2)$
and $c=1+\rho\|\qh_R\|^2$.}
\end{theorem}

\begin{proof}
The secrecy rate maximization problem can be written as
\equ{
\max_{\qQ\succeq 0, \mathrm{Tr}(\qQ)=1}\ C_s(\qQ)=\log\frac{1+\rho\qh_R^\dagger\qQ\qh_R}{1+\rho\qh_E^\dagger\qQ\qh_E}\label{SecRateMaxMISO1}
}
which is a fractional program \cite{Dinkelbach} associated with the following parametric problem
\equ{
F(\alpha)=\max_{\qQ\succeq 0, \mathrm{Tr}(\qQ)=1}\ \left[1+\rho\qh_R^\dagger\qQ\qh_R-\alpha(1+\rho\qh_E^\dagger\qQ\qh_E)\right] \label{SecRateMaxMISO3}
}
where $\alpha>0$. Let $\alpha^\circ$ be the unique root
of $F(\alpha)=0$.
According to \cite{Dinkelbach}, the optimal $\qQ$ corresponding to  $F(\alpha^\circ)$ also optimizes (\ref{SecRateMaxMISO1}).
  Based on the fact that $\qh^\dagger\qQ\qh=\mathrm{Tr}(\qQ\qh\qh^\dagger)$ for any vector $\qh$,
we  rewrite the optimization problem (\ref{SecRateMaxMISO3}) as
\equ{
F(\alpha)=\max_{\qQ\succeq 0, \mathrm{Tr}(\qQ)=1}\ \left[1-\alpha+\rho\mathrm{Tr}\{\qQ(\qh_R\qh_R^\dagger-\alpha\qh_E\qh_E^\dagger)\}\right].
\label{SecRateMaxMISO4}
}
By eigen-decomposition $\qh_R\qh_R^\dagger-\alpha\qh_E\qh_E^\dagger=\qU_{\alpha}\qD_{\alpha}\qU_{\alpha}^\dagger$
and letting $\qQ_{\alpha}=\qU_{\alpha}^\dagger\qQ\qU_{\alpha}$, we obtain $\qQ_{\alpha}\succeq 0$, $\mathrm{Tr}(\qQ_{\alpha})=1$,
$\qQ=\qU_{\alpha}\qQ_{\alpha}\qU_{\alpha}^\dagger$. It holds
\equ{
\mathrm{Tr}\{\qQ(\qh_R\qh_R^\dagger-\alpha\qh_E\qh_E^\dagger)\}=\mathrm{Tr}\{\qQ_{\alpha}\qD_{\alpha}\}
=\mathrm{Tr}\{\mathrm{diag}(\qQ_{\alpha})\qD_{\alpha}\}\le \lambda_{\max}(\qh_R\qh_R^\dagger-\alpha\qh_E\qh_E^\dagger).\label{Qeigenvec}
}
Equation (\ref{Qeigenvec}) holds with equality if
$\qQ_{\alpha}$ is diagonal and has a unique nonzero entry (equal to one) corresponding to position of the largest entry in $\qD_{\alpha}$.
In other words, $\qQ$ and $\qh_R\qh_R^\dagger-\alpha\qh_E\qh_E^\dagger$
have the same eigenvectors, and $\qQ$ has rank one.
Thus, it holds $\qQ=\qu_{\alpha, \max}\qu_{\alpha, \max}^\dagger$
where $\qu_{\alpha, \max}$ is the eigenvector associated
with the largest eigenvalue of $\qh_R\qh_R^\dagger-\alpha\qh_E\qh_E^\dagger$.
The largest eigenvalue and the associated eigenvector of $\qh_R\qh_R^\dagger-\alpha\qh_E\qh_E^\dagger$
can be expressed in closed form based on Lemma 2.
In our problem, $\qr=\qh_R$, $\qs=\sqrt{\alpha}\,\qh_E$.
By using Lemma \ref{lem:2}, we now can obtain
\equ{
F(\alpha)=1+\frac{\rho\|\qh_R\|^2}{2}-\left(1+\frac{\rho\|\qh_E\|^2}{2}\right)\alpha+\frac{\rho}{2}
\sqrt{(\|\qh_R\|^2+\alpha\|\qh_E\|^2)^2-4\alpha|\qh_R^\dagger\qh_E|^2}.\label{equFAlpha}
}
$F(\alpha)=0$ has a unique root given in closed form:
\equ{
\alpha^\circ=\frac{b+\sqrt{b^2-4ac}}{2a}.\label{optAlpha}
}
The optimal $\qQ$ is given by $\qQ^\circ=\qe_1\qe_1^\dagger$ where $\qe_1$ is defined
in Lemma \ref{lem:2} where $\qr=\qh_R$, $\qs=\sqrt{\alpha^\circ}\,\qh_E$. The secrecy capacity is given by $C_s=\log\alpha^\circ$.

\end{proof}

Based on Theorem \ref{theo:MISO},   if $\qh_R= \xi\qh_E$ and $|\xi|<1$, then $b=a+c$, $a-c>0$ and further
$C_s=0$. This is consistent with the fact that when the legitimate receiver channel
is a degraded version of the eavesdropper
channel the secrecy capacity is zero.
If $\qh_R= \xi\qh_E$ and $|\xi|>1$, then $b=a+c$, $a-c<0$ and
$C_s=\log((1+\rho|\xi|^2\|\qh_E\|^2)/(1+\rho\|\qh_E\|^2))>0$.
This is consistent with the fact that when the eavesdropper channel
is a degraded version of the legitimate receiver
channel the secrecy capacity is positive.
If $\qh_R\ne \xi\qh_E$, then $b>a+c$ and it always holds that $C_s>0$.
Thus, if $\qh_R\ne \xi\qh_E$, the MISO wiretap channel always has positive secrecy capacity
independent of the channel.

To gain more insight into the secrecy capacity, we consider
the rate at which the secrecy capacity scales with $\log\rho$ as in \cite{Liang4}.
If $\qh_R\ne \xi\qh_E$, then under high SNR, it follows from (\ref{CsMISO}) that
\equ{
C_s(\rho)=\log\rho+\log(\|\qh_R\|^2-|\qh_R^\dagger\qh_E|^2/\|\qh_E\|^2+O({1}/{\rho}))
}
where $O(\cdot)$ is the big-O notation.
The secrecy degree of freedom ($s.d.o.f.$) (also see \cite{Liang4}) of the MISO wiretap channel is given by
\equ{
s.d.o.f \triangleq \lim_{\rho\to \infty}\frac{C_s(\rho)}{\log\rho}=1.\label{sdofMISO}
}

\section{Conditions for Optimal Input Covariance Matrix of MIMO Wiretap Channel}\label{sec:4}

In this section, we analyze a general MIMO wiretap channel.
First, we obtain the necessary conditions for the optimal $\qQ$ by using Karush-Kuhn-Tucker (KKT) conditions.
Let us construct the cost function
\equ{
L(\qQ, \theta, {\bs \Psi})=C_s(\qQ)-\theta (\mathrm{Tr}(\qQ)-1)+\mathrm{Tr}({\bs \Psi}\qQ)
}
where $\theta$ is the Lagrange multiplier associated with the constraint $\mathrm{Tr}(\qQ)=1$,
${\bs \Psi}$ is the Lagrange multiplier associated with the constraint $\qQ\succeq 0$.
The KKT conditions enable us to write \cite{Boyd}
\al{
&{\bs \Theta}-\theta\qI_{n_T}+{\bs\Psi}=0,\label{KKTcond1}\\
&{\bs \Psi}\succeq 0, \mathrm{Tr}({\bs \Psi}\qQ)=0, \qQ\succeq 0, \mathrm{Tr}(\qQ)=1,\label{KKTcond2}
}
where
\equ{
{\bs\Theta}=\rho\qH_R^\dagger(\qI_{n_R}+\rho\qH_R\qQ\qH_R^\dagger)^{-1}\qH_R-
\rho\qH_E^\dagger(\qI_{n_E}+\rho\qH_E\qQ\qH_E^\dagger)^{-1}\qH_E.\label{Theta}
}
Here we use the facts: $\frac{\partial }{\partial \qQ}\mathrm{Tr}({\bs\Psi}\qQ)={\bs\Psi}^T$ and
\equ{\label{difflogdet}
\frac{\partial \log\mathrm{det}(\qI_{n_R}+\rho\qH_R\qQ\qH_R^\dagger)}{\partial \qQ}
=[\rho\qH_R^\dagger(\qI_{n_R}+\rho\qH_R\qQ\qH_R^\dagger)^{-1}\qH_R]^T.
}
For future use, we also rewrite (\ref{Theta}) as
\equ{
{\bs\Theta}=\qS_R(\qI_{n_T}+\qQ\qS_R)^{-1}-\qS_E(\qI_{n_T}+\qQ\qS_E)^{-1}\label{Theta2}
}
where $\qS_R=\rho\qH_R^\dagger\qH_R$, $\qS_E=\rho\qH_E^\dagger\qH_E$,
which follows from the fact: $\qH(\qI+\rho\qQ\qH^\dagger\qH)=(\qI+\rho\qH\qQ\qH^\dagger)\qH$ for any matrix $\qH$
and hence $(\qI+\rho\qH\qQ\qH^\dagger)^{-1}\qH=\qH(\qI+\rho\qQ\qH^\dagger\qH)^{-1}$,
$\qH^\dagger(\qI+\rho\qH\qQ\qH^\dagger)^{-1}\qH=\qH^\dagger\qH(\qI+\rho\qQ\qH^\dagger\qH)^{-1}$.

In this paper, ${\bs\Theta}$ is an important variable for the optimal input covariance problem.
It has the following property which will used later.
\begin{property}\label{prop:Theta}
{\it For any $\qQ\succeq 0$, $\lambda_{\max}(\bs\Theta)>0$; For any $\qQ\succeq 0$,
$\mathrm{Tr}(\qQ{\bs \Theta})\le \mathrm{Tr}(\qQ)\lambda_{\max}({\bs\Theta})$,
and in particular, for any $\qQ\in \Omega$,
$\mathrm{Tr}(\qQ{\bs \Theta})\le \lambda_{\max}({\bs\Theta})$.}
\end{property}

The proof is given in Appendix \ref{proofPropTheta}.

From the KKT conditions (\ref{KKTcond1}) and (\ref{KKTcond2}), we obtain the equivalent
(but without containing the Lagrange multipliers) conditions for optimal $\qQ$ consisting of a set of
equations given in the following theorem.
\begin{theorem}\label{theo:1}
{\it The optimal $\qQ\succeq 0$ satisfies
\al{
\qQ{\bs \Theta}&=\mathrm{Tr}(\qQ{\bs \Theta})\qQ\label{basicEq1}\\
\lambda_{\max}({\bs\Theta})&=\mathrm{Tr}(\qQ{\bs \Theta}).\label{basicEq2}
}
}
\end{theorem}

Please see Appendix \ref{proofTheo1} for details.

Equations (\ref{basicEq1}) and (\ref{basicEq2}) provide two elementary conditions that characterize the optimal $\qQ$.
At this point we do not have a proof that any $\qQ$ satisfying the conditions of Theorem \ref{theo:1}
is the optimal input covariance.
However, for some special cases, e.g., the MISO wiretap channel analyzed in \S\ref{sec:3},
this is true. In particular, for this case we provide the following theorem.
\begin{theorem}\label{theo:MISOopt}
{\it For MISO wiretap channel, any $\qQ$ satisfying the conditions of Theorem \ref{theo:1} is the optimal input covariance.}
\end{theorem}

The proof is given in Appendix \ref{proofTheoMISOopt}.

Now we proceed. From Property \ref{prop:Theta} and (\ref{basicEq2}), we know that the optimal $\qQ$ satisfies
\equ{
\mathrm{Tr}(\qQ{\bs \Theta})>0\label{basicEq3}.
}
Based on (\ref{basicEq1}) and (\ref{basicEq3}), and
by taking trace operation over both side of (\ref{basicEq1}), it can be easily seen that
$\mathrm{Tr}(\qQ)=1$. That is to say, equations
(\ref{basicEq1}) and (\ref{basicEq2}) imply $\mathrm{Tr}(\qQ{\bs \Theta})>0$ and $\mathrm{Tr}(\qQ)=1$.

The condition (\ref{basicEq1}) reveals that
the optimal $\qQ$ satisfies that $\qQ$ and ${\bs \Theta}$ commute and have the same eigenvectors \cite[p.239]{Davis}.
The condition (\ref{basicEq2}) means that
the eigenvalues of ${\bs\Theta}$ corresponding to the positive eigenvalues of $\qQ$ are all equal to
 $\mathrm{Tr}(\qQ{\bs \Theta})$, while the remaining eigenvalues of ${\bs\Theta}$ (i.e., corresponds to the zero eigenvalues of $\qQ$)
are all less than  or equal to $\mathrm{Tr}(\qQ{\bs \Theta})$. Obviously, if the optimal $\qQ$ has full rank,
then ${\bs \Theta}=\theta\qI_{n_T}$ for a certain $\theta>0$.

It can be shown that based on the conditions of  Theorem \ref{theo:1}, the optimal ${\bf Q}$ has the following properties.
\begin{property}\label{prop:1}
{\it The optimal $\qQ$ satisfies:
\begin{enumerate}
    \item [(i)] $\mathrm{rank}(\bs\Theta)=\mathrm{rank}(\qS_R-\qS_E)\ge \mathrm{rank}(\qQ)$;
    \item [(ii)] $\qQ(\qS_R-\qS_E)\qQ \succeq 0$;
    \item [(iii)] $\qQ+\qQ\qS_E\qQ$ and $\qQ+\qQ\qS_R\qQ$ commute and have the same eigenvectors.
\end{enumerate}
}
\end{property}

For readability, we put the proof of Property \ref{prop:1} in Appendix \ref{proofProp1}.
A direct result of Property \ref{prop:1} is the following:
\begin{property}\label{prop:2}
{\it when $\qS_R-\qS_E\nsucc 0$, the optimal $\qQ$ has low rank.}
\end{property}
The proof is simple. When $\qS_R-\qS_E$ is indefinite, if the optimal $\qQ$ has full rank,
then Property \ref{prop:1} (ii) leads to $\qS_R-\qS_E\succeq 0$,
which violates that $\qS_R-\qS_E$ is indefinite.
When $\qS_R-\qS_E\succeq 0$ but $\qS_R-\qS_E\nsucc 0$, it follows from Property \ref{prop:1} (i) that the optimal $\qQ$ has low rank.
This result was also pointed out in \cite{Hassibi2}. When $\qS_R-\qS_E\succ 0$,
the optimal $\qQ$ may have low rank or full rank.

Before ending this section, we point out that we can
combine the elementary conditions (\ref{basicEq1}) and (\ref{basicEq2}) into a single equation.
When ${\bs \Theta}$ and $\qQ$ commute and have the same eigenvectors,
${\bs \Theta}+\gamma \qI_{n_T}$ and $\qQ$ commute
and have the same eigenvectors for any real number $\gamma$, and vice versa.
We can find a certain $\gamma$ such that ${\bs \Theta}+\gamma \qI_{n_T}\succ 0$ for any $\qQ\succeq 0$.
Based on $\bs\Theta$ in (\ref{Theta2}), we have
\equ{
\lambda_{\min}({\bs\Theta})> -\lambda_{\max}(\qS_E).
}
When $\gamma \ge \lambda_{\max}(\qS_E)$, it  always holds that ${\bs \Theta}+\gamma \qI_{n_T}\succ 0$.
Let $\qK={\bs \Theta}+\gamma \qI_{n_T}$ and hence $\qK \succ 0$, $\mathrm{Tr}(\qQ\qK)>0$ for any $\qQ\succeq 0$ but $\qQ\ne 0$.
Equations (\ref{basicEq1}) and (\ref{basicEq2})
are equivalent to
\al{
\qQ\qK&=\mathrm{Tr}(\qQ\qK)\qQ, \label{algoBasicEq1}\\
\lambda_{\max}(\qK)&=\mathrm{Tr}(\qQ\qK).\label{algoBasicEq2}
}
We can combine the above two equations to a single one as follows.
\equ{
\qQ\qK=\frac{1}{2}\left(\mathrm{Tr}(\qQ)+\frac{1}{\mathrm{Tr}(\qQ)}\right)\lambda_{\max}(\qK)\qQ\label{singleEq}
}
or equivalently,
\equ{
\qQ\wt{\qK}=\frac{1}{2}\left(\mathrm{Tr}(\qQ)+\frac{1}{\mathrm{Tr}(\qQ)}\right)\qQ\label{singleEq2}
}
where $\wt{\qK}=\qK/\lambda_{\max}(\qK)$.
We give the following theorem.
\begin{theorem}\label{theo:singleEq}
{\it Any $\qQ\succeq 0$ that satisfies (\ref{singleEq}) also satisfies the conditions of Theorem \ref{theo:1}.}
\end{theorem}
\begin{proof}
It is easy to verify that $\frac{1}{2}(\mathrm{Tr}(\qQ)+1/\mathrm{Tr}(\qQ))\lambda_{\max}(\qK)$
is an eigenvalue of $\qK$. But $\lambda_{\max}(\qK)$ is the largest eigenvalue of $\qK$.
On the other hand, $\frac{1}{2}(\mathrm{Tr}(\qQ)+1/\mathrm{Tr}(\qQ))\ge 1$ holds with equality
if and only if $\mathrm{Tr}(\qQ)=1$. Thus, we know $\mathrm{Tr}(\qQ)=1$. With this, taking trace operation over
both sides of (\ref{singleEq}) leads to $\lambda_{\max}(\qK)=\mathrm{Tr}(\qQ\qK)$.
\end{proof}

Summarily, we can alternatively do one of the following two things to find $\qQ$ that satisfies the conditions of Theorem \ref{theo:1}:
\begin{enumerate}
    \item [(i)] Find $\qQ\succeq 0$ satisfies (\ref{algoBasicEq1}) and (\ref{algoBasicEq2});
    \item [(ii)] Find $\qQ\succeq 0$ satisfies (\ref{singleEq});
\end{enumerate}
We will discuss the algorithm to search for such $\qQ$ in \S\ref{algori}.

In the following sections, we will analyze some special cases.
In particular, for $n_T=2$ we obtain the optimal $\qQ$ in closed form.
If $\qH_R^\dagger\qH_R-\qH_E^\dagger\qH_E$ has all negative eigenvalues except
a positive eigenvalue, we show that the optimal $\qQ$ has rank one and can also can be expressed in a closed form.
For general cases, we prove that the optimal $\qQ$ is a fixed point of a mapping from a convex set to itself,
and propose an algorithm to search for it.

\section{The Case $n_T=2$}\label{nT2}

In this section, we analyze the case $n_T=2$, i.e., the transmitter has two antennas.
It includes the four cases $(n_T, n_R, n_E)=(2,2,2), (2,2,1), (2,1,2), (2,1,1)$.
In \S\ref{sec:3}, the MISO wiretap channel with $n_T=2$ belongs to $(n_T, n_R, n_E)=(2,1,1)$.
In \cite{Shafiee}, the case $(n_T, n_R, n_E)=(2,2,1)$ is analyzed.
We derive the optimal $\qQ$ in two subsections in which $\qS_R-\qS_E\nsucc 0$ or $\qS_R-\qS_E\succ 0$.
We also analyze the rank of optimal $\qQ$ with respect to SNR.

\subsection{$\qS_R-\qS_E\nsucc 0$}\label{nT2indef}

According to Property \ref{prop:2}, the optimal $\qQ$ has low rank (rank one)
and hence it has the form $\qQ=\qu\qu^\dagger$ where $\qu$ is a unit-norm vector to be determined.
We can rewrite
\equ{
C_s(\qQ)=\log\frac{1+\qu^\dagger\qS_R\qu}{1+\qu^\dagger\qS_E\qu}
=\log\frac{\qu^\dagger(\qI_2+\qS_R)\qu}{\qu^\dagger(\qI_2+\qS_E)\qu}.\label{optQnT2indef}
}
The optimal $\qQ$ is easily obtained to be $\qQ^\circ=\qu^\circ{\qu^\circ}^\dagger$
where $\qu^\circ$ is the eigenvector associated with the largest eigenvalue of $(\qI_2+\qS_E)^{-1}(\qI_2+\qS_R)$.
The secrecy capacity is given by
\equ{
C_s=\log\left(\lambda_{\max}\{(\qI_2+\qS_E)^{-1}(\qI_2+\qS_R)\}\right).
}
We can express $C_s$ in closed form. Denote
\equ{
\qH_E^\dagger\qH_E=\left(%
\begin{array}{cc}
  a_1 & b_1 \\
  b_1^\ast & c_1 \\
\end{array}%
\right), \qH_R^\dagger\qH_R=\left(%
\begin{array}{cc}
  a_2 & b_2 \\
  b_2^\ast & c_2 \\
\end{array}%
\right).\label{S1S2}
}
By using the fact: for any $2\times 2$ matrix $\qA$ with two real eigenvalues, the largest eigenvalue is given by
$\lambda_{\max}(\qA)=[\mathrm{Tr}(\qA)+\sqrt{(\mathrm{Tr}(\qA))^2-4\det(\qA)}\,]/2$, and the matrix inverse formula
\equ{
\left(%
\begin{array}{cc}
  a_{11} & a_{12} \\
  a_{12}^\ast & a_{22} \\
\end{array}%
\right)^{-1}=\frac{1}{a_{11}a_{22}-|a_{12}|^2} \left(%
\begin{array}{cc}
  a_{22} & -a_{12} \\
  -a_{12}^\ast & a_{11} \\
\end{array}%
\right),\label{matrixinv}
}
we can obtain
\equ{
C_s=\log\frac{A+\sqrt{A^2-4(B_1-B_2)}}{2[(1/\rho+a_1)(1/\rho+c_1)-|b_1|^2]}\label{CsClosedFrom}
}
where $A=a_1c_2+a_2c_1-b_1b_2^\ast-b_1^\ast b_2+(a_1+a_2+c_1+c_2)/\rho+2/\rho^2$,
$B_1=(a_2c_1-b_1b_2^\ast+(a_2+c_1)/\rho+1/\rho^2)(a_1c_2-b_1^\ast b_2+(a_1+c_2)/\rho+1/\rho^2)$ and
$B_2=(b_2c_1-b_1c_2+(b_2-b_1)/\rho)(a_1b_2^\ast-a_2b_1^\ast+(b_2^\ast-b_1^\ast)/\rho)$.

Now we analyze the secrecy degree of freedom (s.d.o.f.) defined in (\ref{sdofMISO}) which is different whether $\qS_E$ has full rank or not.
\begin{itemize}

\item \underline{Case 1) $\qS_E$ has rank two (full rank)}

In this case, noting that
$(\qI_2+\qS_E)^{-1}(\qI_2+\qS_R)=(\qI_2/\rho+\qH_E^\dagger\qH_E)^{-1}(\qI_2/\rho+\qH_R^\dagger\qH_R)$,
we have $C_s\to \log(\lambda_{\max}\{(\qH_E^\dagger\qH_E)^{-1}\qH_R^\dagger\qH_R\})$ as $\rho\to \infty$.
Thus, we get
\equ{
s.d.o.f = \lim_{\rho\to \infty}\frac{C_s(\rho)}{\log\rho} = 0.
}

\item \underline{Case 2) $\qS_E$ has rank one (low rank)}

In this case, $\qH_E^\dagger\qH_E$ is singular, hence can be expressed as $\qH_E^\dagger\qH_E=\qv_2\qv_2^\dagger$.
By using (\ref{matrixinverse}), we can write
$(\qI_2+\rho\qS_2)^{-1}(\qI_2+\rho\qS_1)=\rho(\qI_2-\rho\qv_2\qv_2^\dagger/(1+\rho\|\qv_2\|^2))(\qI_2/\rho+\qS_1)
\to \rho(\qI_2-\qS_2/\mathrm{Tr}(\qS_2))\qS_1$ as $\rho\to \infty$.
Thus, as $\rho\to \infty$, $\lambda_{\max}((\qI_2+\rho\qS_2)^{-1}(\qI_2+\rho\qS_1))
\to \rho \lambda_{\max}((\qI_2-\qS_2/\mathrm{Tr}(\qS_2))\qS_1)$.
Thus, we get
\equ{
s.d.o.f = \lim_{\rho\to \infty}\frac{C_s(\rho)}{\log\rho} = 1.
}
We can also use (\ref{CsClosedFrom}) to obtain the same result.

\end{itemize}

\subsection{$\qS_R-\qS_E\succ 0$}

In this case, the  optimal $\qQ$ may have full rank or low rank.
If the optimal $\qQ$ has low rank, it is given in (\ref{optQnT2indef}). Therefore, in the following we focus on the
case in which the optimal $\qQ$ has full rank.
The optimal $\qQ$ can be determined from the above two cases.

Since $\qQ\succ 0$, it follows from (\ref{basicEq1})
that ${\bs\Theta}$ must be a positive scalar multiplication of $\qI_2$.
Recall from (\ref{Theta2}) that ${\bs\Theta}=\qS_R(\qI_2+\qQ\qS_R)^{-1}-\qS_E(\qI_2+\qQ\qS_E)^{-1}$.
We know $\qS_R\succ 0$, but $\qS_E$ is not necessarily positive definite.
Thus, in the following, we discuss two cases respectively:
a) $\qS_E$ has rank two (full rank); b) $\qS_E$ has rank one (low rank).

\begin{itemize}
\item \underline{Case a) $\qS_E$ has rank two (full rank)} In this case, $\qS_R\succ 0$, $\qS_E \succ 0$. We can rewrite
\equ{
{\bs\Theta}=(\qS_R^{-1}+\qQ)^{-1}-(\qS_E^{-1}+\qQ)^{-1}=\theta\qI_2, \theta>0.\label{Qpositi}
}
Based on the eigen-decomposition $(\qS_E^{-1}-\qS_R^{-1})/2=\qU_1\qD_1\qU_1^\dagger$
where $\qD_1=\mathrm{diag}(d_1,d_2)$, $d_1\ge d_2\ge0$ but $d_1^2+d_2^2\ne 0$ (otherwise, $\qS_R=\qS_E$ violates the assumption $\qS_R\ne\qS_E$),
and letting $\qC=\qU_1^\dagger(\qS_E^{-1}+\qS_R^{-1})\qU_1/2$,
we  get $\qS_E^{-1}=\qU_1(\qC+\qD_1)\qU_1^\dagger$ and $\qS_R^{-1}=\qU_1(\qC-\qD_1)\qU_1^\dagger$. On inserting the latter expressions in
(\ref{Qpositi}) we get
\equ{
(\hat{\qQ}+\qC-\qD_1)^{-1}-(\hat{\qQ}+\qC+\qD_1)^{-1}=\theta\qI_2, \theta>0\label{Qpositi1}
}
where $\hat{\qQ}=\qU_1^\dagger\qQ\qU_1$. Note that $\qD_1$ is diagonal. We can actually show that $\hat{\qQ}+\qC$ must be diagonal.
To prove this, let us denote the $(1,2)$th entry of $\hat{\qQ}+\qC$ by $\bar{q}_{12}$.
We know that the $(1,2)$th entry of $(\hat{\qQ}+\qC-\qD_1)^{-1}-(\hat{\qQ}+\qC+\qD_1)^{-1}$
equals
\equ{
-\bar{q}_{12}[\mathrm{det}(\hat{\qQ}+\qC-\qD_1)^{-1}-\mathrm{det}(\hat{\qQ}+\qC+\qD_1)^{-1}]=0
}
which leads to $\bar{q}_{12}=0$. Here we used (\ref{matrixinv}) and
the fact $\det(\qA)>\det(\qB)$ for $\qA\succ \qB$, $\qB\succ 0$.
Since $\hat{\qQ}+\qC$ is diagonal, we denote $\hat{\qQ}+\qC=\mathrm{diag}(\bar{q}_1,\bar{q}_2)$ and
$\hat{\qQ}=\mathrm{diag}(\bar{q}_1,\bar{q}_2)-\qC$. It follows from $\hat{\qQ}\succ 0$, $\mathrm{Tr}(\hat{\qQ})=1$
that $\bar{q}_1+\bar{q}_2=1+\mathrm{Tr}(\qC)$, $\mathrm{diag}(\bar{q}_1,\bar{q}_2)\succ\qC$.
Combining these with (\ref{Qpositi1}) results in
\equ{
(d_1/d_2-1)\bar{q}_2^2+2(1+\mathrm{Tr}(\qC))\bar{q}_2+d_1(d_1-d_2)-(1+\mathrm{Tr}(\qC))^2=0.\label{quadequ}
}
We can solve $\bar{q}_2$ from the quadratic equation (\ref{quadequ}) and $\bar{q}_1=1+\mathrm{Tr}(\qC)-\bar{q}_2$.
If $\mathrm{diag}(\bar{q}_1,\bar{q}_2)\succ\qC$ holds, then
$\qQ=\qU_1\hat{\qQ}\qU_1^\dagger$ is a possible solution.
If the equation (\ref{quadequ}) has no positive roots or $\mathrm{diag}(\bar{q}_1,\bar{q}_2)\nsucc\qC$, it means the optimal
$\qQ$ has low rank.

\item \underline{Case b) $\qS_E$ has rank one (low rank)}
In this case $\qS_E$ can be expressed as $\qS_E=\qv_e\qv_e^\dagger$.
We eigen-decompose $\qS_E=\qU_e\mathrm{diag}(\lambda_e,0)\qU_e^\dagger$.
Similarly, we get
\equ{
{\bs\Theta}=(\qS_R^{-1}+\qQ)^{-1}-\qS_E(\qI_2+\qQ\qS_E)^{-1}=\theta\qI_2, \theta>0.\label{Qpositi2}
}
Let us define $\breve{\qQ}=\qU_e^\dagger\qQ\qU_e$, $\breve{\qS}_R=\qU_e^\dagger\qS_R^{-1}\qU_e$.
Inserting $\breve{\qQ}$ and $\breve{\qS}_R$ into (\ref{Qpositi2})
results in
\equ{
(\breve{\qS}_R+\breve{\qQ})^{-1}-\frac{\lambda_e}{1+\breve{q}_{11}\lambda_e}\left(%
\begin{array}{cc}
  1 & 0 \\
  0 & 0 \\
\end{array}%
\right)=\theta\qI_2, \theta>0.\label{Qpositi3}
}
where $\breve{q}_{11}$ is the $(1,1)$th entry of $\breve{\qQ}$. It follows from (\ref{Qpositi3}) that $\breve{\qS}_R+\breve{\qQ}$
is diagonal. Thus, we denote $\breve{\qS}_R+\breve{\qQ}=\mathrm{diag}(\tilde{q}_1,\tilde{q}_2)$
and $\breve{\qQ}=\mathrm{diag}(\tilde{q}_1,\tilde{q}_2)-\breve{\qS}_R$. It follows from $\breve{\qQ}\succ 0$, $\mathrm{Tr}(\breve{\qQ})=1$
that $\tilde{q}_1+\tilde{q}_2=1+\mathrm{Tr}(\breve{\qS}_R)$ and $\mathrm{diag}(\tilde{q}_1,\tilde{q}_2)\succ \breve{\qS}_R$.
Combining the above and (\ref{Qpositi3}) results in
\equ{
\lambda_e\tilde{q}_1^2+2(1-(\breve{\qS}_R)_{11}\lambda_e)\tilde{q}_1-(1-
(\breve{\qS}_R)_{11}\lambda_e)(1+\mathrm{Tr}(\breve{\qS}_R))=0\label{quadeq2}
}
where $(\breve{\qS}_R)_{11}$ is the $(1,1)$th entry of $\breve{\qS}_R$.
We can solve $\tilde{q}_1$ from the quadratic equation (\ref{quadeq2}), and then get
$\tilde{q}_2=1+\mathrm{Tr}(\breve{\qS}_R)-\tilde{q}_1$.
If $\mathrm{diag}(\tilde{q}_1,\tilde{q}_2)\succ \breve{\qS}_R$ holds, then
$\qQ=\qU_e\breve{\qQ}\qU_e^\dagger$ is a possible solution.
If the equation (\ref{quadeq2}) has no positive roots or $\mathrm{diag}(\tilde{q}_1,\tilde{q}_2)\nsucc \breve{\qS}_R$, it means the optimal
$\qQ$ has low rank.
\end{itemize}

\subsection{Rank of Optimal $\qQ$}

For the non-wiretap MIMO channel the rank
of optimal input covariance has a non-decreasing property with respect to SNR \cite{Raghavan}.
In this section we consider the behavior of the rank of optimal input covariance of the MIMO wiretap channel
with respect to SNR.

When $\qS_R-\qS_E\nsucc 0$, according to the result in \S\ref{nT2indef},
the optimal $\qQ$ has rank one, independent of SNR, and hence follows the non-decreasing property of rank.
Next we focus on $\qS_R-\qS_E\succ 0$.
According to \S\ref{nT2indef}, if the optimal $\qQ$ has rank one, it can be expressed as
$\qQ=\qu_0\qu_0^\dagger$ where $\qu_0$ is the eigenvector associated with the largest eigenvalue $\lambda_0$ of
$(\qI_2+\rho\qS_2)^{-1}(\qI_2+\rho\qS_1)$ where $\qS_1=\qH_R^\dagger\qH_R$, $\qS_2=\qH_E^\dagger\qH_E$.
Denote $\qS_0=\qS_1(\qI_2+\rho\qu_0\qu_0^\dagger\qS_1)^{-1}-
\qS_2(\qI_2+\rho\qu_0\qu_0^\dagger\qS_2)^{-1}$. Since $\qS_1-\qS_2\succ 0$, we can rewrite
$\qS_0=(\qI_2+\rho\qS_1\qu_0\qu_0^\dagger)^{-1}(\qS_1-\qS_2)(\qI_2+\rho\qu_0\qu_0^\dagger\qS_2)^{-1}$
and state that $\qS_0$ has full rank (rank two).
Denote
\equ{
g(\rho)\triangleq 2\qu_0^\dagger\qS_0\qu_0-\mathrm{Tr}(\qS_0)
=\frac{\qu_0^\dagger(2\qS_1+\rho\qS_1^2)\qu_0}{\qu_0^\dagger(\qI_{n_T}+\rho\qS_1)\qu_0}
-\frac{\qu_0^\dagger(2\qS_2+\rho\qS_2^2)\qu_0}{\qu_0^\dagger(\qI_{n_T}+\rho\qS_2)\qu_0}-\mathrm{Tr}(\qS_1-\qS_2)
}
where the matrix inverse formula (\ref{matrixinverse}) is used.
We give the following result.
\begin{lemma}\label{lem:rankQ}
{\it If $g(\rho)<0$, then the optimal $\qQ$ has rank two;
If the optimal $\qQ$ has rank one, there must be $g(\rho)\ge 0$.
}
\end{lemma}

Please see Appendix \ref{proofLemRankQ} for details.

Then we can prove the following result.
\begin{theorem}\label{theo:rankQ}
{\it $\lim_{\rho\to\infty} g(\rho)<0$, hence, according to Lemma \ref{lem:rankQ}, there exists a certain $\rho_0$ such that when $\rho>\rho_0$,
the optimal $\qQ$ has rank two.
}
\end{theorem}

The proof is given in Appendix \ref{proofTheoRankQ}.

Theorem \ref{theo:rankQ} reveals that when the SNR is sufficient large, the optimal $\qQ$ always has rank two.
At this point, we do not prove the rank non-decreasing property of the optimal $\qQ$ for the case $\qS_R-\qS_E\succ 0$ completely.

\section{$\qS_R-\qS_E$ Has All Negative Eigenvalues Except
One Positive Eigenvalue}\label{OnePosiEig}

We analyze the case in which $\qH_R^\dagger\qH_R-\qH_E^\dagger\qH_E$ has all negative eigenvalues except
one positive eigenvalue, e.g., $n_T=3$, $\qH_R^\dagger\qH_R-\qH_E^\dagger\qH_E$ has two negative eigenvalues
and a positive eigenvalue.
In particular, this always occurs when $\qS_E$ has full rank and $n_R=1$, i.e., the legitimate receiver has a single antenna,
as the following lemma stated.

\begin{lemma}\label{lem:posi}
{\it Let $\qx$ is a known non-zero vector, $\qX$ is a known positive semi-definite matrix.
Assume $\qx\qx^\dagger-\qX$ is indefinite or positive semi-definite.
If $\qX$ has full rank, or $\qX$ has rank $n_T-1$ and $\qx$ is linearly independent of the eigenvectors associated with
the non-zero eigenvalues of $\qX$,
then $\qx\qx^\dagger-\qX$ has all negative eigenvalues except one positive eigenvalue.}
\end{lemma}
\begin{proof}
First, we prove the case that $\qX$ has full rank.
Let $\lambda$ be any eigenvalue of $\qx\qx^\dagger-\qX$.
It holds $\det(\lambda\qI-\qx\qx^\dagger+\qX)=0$.
When $\lambda\ge 0$, noting that $\lambda\qI+\qX\succ 0$,
we get $\det(\lambda\qI-\qx\qx^\dagger+\qX)
=\det(\lambda\qI+\qX)\det(\qI-(\lambda\qI+\qX)^{-1}\qx\qx^\dagger)
=\det(\lambda\qI+\qX)(1-\qx^\dagger(\lambda\qI+\qX)^{-1}\qx)=0$ which leads to
$\qx^\dagger(\lambda\qI+\qX)^{-1}\qx=1$.
Here we use the fact
$\det(\qI+\qA\qB)=\det(\qI+\qB\qA)$.
It is easy to prove that
$\qx^\dagger(\lambda\qI+\qX)^{-1}\qx$ decreases strictly with $\lambda$.
Thus, there is at most one $\lambda\ge 0$ such that $\qx^\dagger(\lambda\qI+\qX)^{-1}\qx=1$.
This, when combined with the fact that $\qx\qx^\dagger-\qX$ is indefinite or positive semi-definite,
gives the desired result.

Second, we prove the case that $\qX$ has rank $n_T-1$ and $\qx$ is linearly independent of the eigenvectors associated with
the non-zero eigenvalues of $\qX$.
Denote the eigen-decomposition $\qX=\lambda_1\qu_1\qu_1^\dagger+\cdots+\lambda_{n_T-1}\qu_{n_T-1}\qu_{n_T-1}^\dagger$.
We can write $\qx\qx^\dagger-\qX=\qF_x\qD_x\qF_x^\dagger$ where $\qF_x=[\qx, \qu_1, \cdots, \qu_{n_T-1}]$,
$\qD_x=\mathrm{diag}(1,-\lambda_1,\cdots,-\lambda_{n_T-1})$.
Since $\qx$ is linearly independent of the eigenvectors associated with
the non-zero eigenvalues of $\qX$, it holds that $\qF_x$ has full rank.
According to Sylvester's law of inertia \cite[p.223]{Horn}, we know $\qF_x\qD_x\qF_x^\dagger$ and $\qD_x$
have the same number of positive, negative, and zero eigenvalues,
thus the desired result is obtained.\\
\end{proof}

According to the above lemma:
when $n_R=1$ (hence, $\qS_R$ can be expressed $\qS_R=\qv_r\qv_r^\dagger$) and $\qS_E$ has full rank,
$\qS_R-\qS_E$ has all negative eigenvalues except one positive eigenvalue;
when $n_R=1$, $\qS_E$ has rank $n_T-1$ and $\qv_r$ is linearly independent of the eigenvectors associated with
the non-zero eigenvalues of $\qS_E$,
$\qS_R-\qS_E$ has all negative eigenvalues except one positive eigenvalue.
But we point out that it does not limit to the cases in Lemma \ref{lem:posi} in which
$\qH_R^\dagger\qH_R-\qH_E^\dagger\qH_E$ has all negative eigenvalues except
one positive eigenvalue. In fact, this will even occur when $\qS_R$ and $\qS_E$ both have full rank.
When $\qH_R^\dagger\qH_R-\qH_E^\dagger\qH_E$ has all negative eigenvalues except
one positive eigenvalue, we give the following theorem.
\begin{theorem}\label{theo:2}
{\it If $\qH_R^\dagger\qH_R-\qH_E^\dagger\qH_E$ has all negative eigenvalues except
one positive eigenvalue, the optimal $\qQ$ has rank one. Also, the optimal $\qQ$
is given by $\qQ^\circ=\qu^\circ{\qu^\circ}^\dagger$
where $\qu^\circ$ is the eigenvector associated with the largest eigenvalue of
$(\qI_{n_T}+\qS_E)^{-1}(\qI_{n_T}+\qS_R)$.
The secrecy capacity is given by
\equ{
C_s=\log\left(\lambda_{\max}\{(\qI_{n_T}+\qS_E)^{-1}(\qI_{n_T}+\qS_R)\}\right).
}
}
\end{theorem}

Please see Appendix \ref{proofTheo2} for details.

Similar to \S\ref{nT2indef}, if $\qH_E^\dagger\qH_E$ has full rank, we obtain
\equ{
s.d.o.f = \lim_{\rho\to \infty}\frac{C_s(\rho)}{\log\rho} = 0
}
and if $\qH_E^\dagger\qH_E$ has low rank, we obtain
\equ{
s.d.o.f = \lim_{\rho\to \infty}\frac{C_s(\rho)}{\log\rho} = 1.
}

\section{Algorithm for General MIMO Wiretap Channel}\label{algori}

In this section, we propose an algorithm to search for the optimal $\qQ$
which applies for any MIMO wiretap channel.
The algorithm is based on the conditions of Theorem \ref{theo:1} or Theorem \ref{theo:singleEq}.

It follows from (\ref{algoBasicEq1}) that
\equ{
\qK^{1/2}\qQ\qK^{1/2}=\mathrm{Tr}(\qQ\qK)\qQ
}
which enables us to get
\equ{
\qQ=\frac{\qK^{1/2}\qQ\qK^{1/2}}{\mathrm{Tr}(\qQ\qK)}
\triangleq f(\qQ).\label{mapping2}
}
Note that $f(\qQ)\succeq 0$ and $\mathrm{Tr}(f(\qQ))=1$ for any $\qQ\in\Omega$.
The equation (\ref{mapping2}) defines a mapping from a convex set to itself: $\Omega \to \Omega$, $\qQ \mapsto f(\qQ)$.
The optimal $\qQ$ corresponds to a fixed point of $f(\qQ)$, i.e., $f(\qQ^\circ)=\qQ^\circ$.
To search for the fixed point, the iterative expression is
\equ{
\qQ^{k+1}=f(\qQ^{k}), k=0,1,\cdots
}
The initial point $\qQ^0$ can be set to $\qI_{n_T}$ or choose a good initial point.
The iterations stop when $\|\qQ^{k+1}-\qQ^k\|<10^{-6}$. If the convergent $\qQ$ satisfies (\ref{algoBasicEq2}),
we obtain a solution satisfying the conditions of Theorem \ref{theo:1},
otherwise, we choose a different initial point.

\section{Numerical Simulations}\label{sec:sim}

We give some examples to illustrate the proposed algorithm.
For illustration purpose, we consider a MIMO wiretap channel where $n_T=4$, $n_R=4$, $n_E=3$.

First, we take an example for $\qH_R^\dagger\qH_R-\qH_E^\dagger\qH_E\succ 0$. The channel matrices are given by
\equ{
\qH_R=\left(%
\begin{array}{cccc}
  -0.1107 - 0.1225i & 0.0582 - 0.3483i & 0.3239 - 0.0071i & -0.2872 - 0.2655i \\
  0.5128 - 0.3239i & -0.8903 - 0.0318i & -0.5524 - 0.0365i & -0.2072 + 0.3047i \\
  -0.0041 + 0.0265i & 0.0871 - 0.0253i & 0.0183 + 1.1679i & -0.0784 + 0.0415i \\
  -0.4699 - 0.1014i & -0.0888 + 0.1127i & 0.2099 + 0.3282i & 0.1734 - 0.4146i \\
\end{array}%
\right)
}
and
\equ{
\qH_E=\left(%
\begin{array}{cccc}
  -0.0766 + 0.1370i & -0.0977 - 0.0985i & 0.0002 - 0.0695i & 0.0583 + 0.0356i \\
  -0.0355 - 0.1167i & 0.1607 - 0.1091i & -0.0809 + 0.1481i & -0.0218 + 0.1109i \\
  0.1375 - 0.0381i & -0.0845 - 0.0610i & -0.0011 + 0.1129i & -0.0393 + 0.1124i \\
\end{array}%
\right).
}
The eigenvalues of $\qH_R^\dagger\qH_R-\qH_E^\dagger\qH_E$ are $0.0085,    0.3704,    0.8945,    2.5213$.
Fig. \ref{fig:1}-\ref{fig:3} depict respectively the eigenvalues of $\qQ^k$, secrecy rate and $\|\qQ^{k+1}-\qQ^k\|$ in the iterations
where the SNR is $\rho=8\, \mbox{dB}$.
Fig. \ref{fig:4}-\ref{fig:5} depict respectively the (possible) secrecy capacity and eigenvalues of (possible) optimal $\qQ$ under different SNRs.
It can be seen from Fig. \ref{fig:5} that the (possible) optimal $\qQ$ can have rank one to four with the increasing SNR,
which shows that
when $\qH_R^\dagger\qH_R-\qH_E^\dagger\qH_E\succ 0$, the (possible) optimal $\qQ$ may have low rank or full rank.

Secondly, we take an example for $\qH_R^\dagger\qH_R-\qH_E^\dagger\qH_E\nsucc 0$. The channel matrices are given by
\equ{
\qH_R=\left(%
\begin{array}{cccc}
  -0.1110 - 0.0667i & -0.1937 - 0.1349i & -0.0752 - 0.2707i & -0.2718 + 0.2730i \\
  0.2877 + 0.6779i & -0.7832 - 0.2249i & 0.4350 + 0.2637i & 0.4160 + 0.5109i \\
  0.3266 - 0.2779i & -0.2345 - 0.4472i & 0.2448 + 0.3488i & -0.6794 - 0.0117i \\
  -0.1221 + 0.4915i & 0.0959 - 0.2557i & -0.0219 + 0.5077i & 0.1449 + 0.3294i \\
\end{array}%
\right)
}
and
\equ{
\qH_E=\left(%
\begin{array}{cccc}
  0.1468 - 0.1185i & 0.4071 + 0.4469i & 0.2474 - 0.3291i & -0.6264 - 0.1313i \\
  -0.0520 + 0.2917i & -0.4978 + 0.0545i & 0.0779 - 0.3472i & -0.0132 - 0.1327i \\
  0.5799 - 0.1767i & 0.2298 + 0.3331i & -0.1151 - 0.2000i & 0.1404 - 0.3501i \\
\end{array}%
\right).
}
The eigenvalues of $\qH_R^\dagger\qH_R-\qH_E^\dagger\qH_E$ are $-0.8206,   -0.1565,    0.9365,    1.8506$.
Figures \ref{fig:6}-\ref{fig:8} depict respectively the eigenvalues of $\qQ^k$, secrecy rate and $\|\qQ^{k+1}-\qQ^k\|$ in the iterations
where the SNR is $\rho=8\, \mbox{dB}$.
Figures \ref{fig:9}-\ref{fig:10} depict respectively the (possible) secrecy capacity and eigenvalues of (possible) optimal $\qQ$ under different SNRs.
It can be seen from Fig. \ref{fig:10} that the (possible) optimal $\qQ$ always has rank two,
which equals the number of positive eigenvalues of $\qH_R^\dagger\qH_R-\qH_E^\dagger\qH_E$.

\section{Conclusion}\label{sec:conclu}

We have investigated the problem of finding the optimal input covariance matrix that achieves secrecy capacity subject to a power constraint.
In particular, for the multiple-input single-output (MISO) wiretap channel, the optimal input covariance matrix is obtained in closed form.
For general cases, we derive the necessary conditions for the optimal solution consisting of a set of equations.
For the case in which the transmitter has two antennas, the derived necessary conditions can result in
a closed form solution.
If the difference is indefinite and has all negative eigenvalues except one positive eigenvalue,
we prove that the optimal input covariance matrix has rank one and can be obtained in closed form.
For other cases, we prove that the solution is a fixed point of a mapping from a convex set to itself and provide an iterative procedure to
search for it.

\appendices

\section{Proof of Lemma \ref{lem:1}}\label{proofLem1}

First we show the necessary part. When $\qH_R^\dagger\qH_R-\qH_E^\dagger\qH_E$ is negative semi-definite,
so is $\qR_x^{1/2}(\qH_R^\dagger\qH_R-\qH_E^\dagger\qH_E)\qR_x^{1/2}$ which leads to
$\qI_{n_T}+\qR_x^{1/2}\qH_E^\dagger\qH_E\qR_x^{1/2}/\sigma^2 \succeq \qI_{n_T}+\qR_x^{1/2}\qH_R^\dagger\qH_R\qR_x^{1/2}/\sigma^2$.
With this, using the fact: if $\qA\succ 0$, $\qB\succ 0$, $\qA \succeq \qB$ then $\det(\qA)\ge \det(\qB)$,
and applying the identity $\mathrm{det}(\qI+\qA\qB)=\mathrm{det}(\qI+\qB\qA)$ to (\ref{SecRate}) results in $C_s(\qR_x)\le 0$.
To show the sufficient part, we rewrite (\ref{SecRate}) as
\equ{
C_s(\qR_x)=\log\mathrm{det}\left(\qI_{n_T}+\frac{1}{\sigma^2}\qR_x^{1/2}(\qH_R^\dagger\qH_R-\qH_E^\dagger\qH_E)\qR_x^{1/2}
(\qI_{n_T}+\qR_x^{1/2}\qH_E^\dagger\qH_E\qR_x^{1/2}/\sigma^2)^{-1}\right).
}
Note that $(\qI_{n_T}+\qR_x^{1/2}\qH_E^\dagger\qH_E\qR_x^{1/2}/\sigma^2)^{-1}\succ 0$,
and it suffices to show that there exists $\qR_x$ such that $\qR_x^{1/2}(\qH_R^\dagger\qH_R-\qH_E^\dagger\qH_E)\qR_x^{1/2}\succeq 0$.
Let us define the eigen-decomposition $\qH_R^\dagger\qH_R-\qH_E^\dagger\qH_E=\qU_r\qD_r\qU_r^\dagger$.
It is easy to verify that $\qR_x=\qU_r\qD_x\qU_r^\dagger$
is a choice where the entries of the diagonal $\qD_x$ are zero corresponding to the position of negative entries in $\qD_r$.

\section{Proof of Property \ref{prop:Theta}}\label{proofPropTheta}

First, we prove the former part.
If $\qS_R\succ 0$ and $\qS_E\succ 0$, then we can rewrite (\ref{Theta2}) as
\equ{
{\bs\Theta}=(\qS_R^{-1}+\qQ)^{-1}-(\qS_E^{-1}+\qQ)^{-1}.
}
We can state that ${\bs\Theta}$ is not negative semi-definite,
otherwise, we get $\qS_E\succeq \qS_R$ which violates the assumption that $\qS_R-\qS_E$
is indefinite or positive semi-definite.

Next we consider the case that $\qS_R$ or $\qS_E$ are singular.
Denote $\qS_R-\qS_E={\bs\Delta}$.
By using the fact: $\qS_E(\qI+\qQ\qS_E)^{-1}=(\qI+\qS_E\qQ)^{-1}\qS_E$,
we can rewrite (\ref{Theta2}) as
\al{
{\bs\Theta}&=\qS_R(\qI_{n_T}+\qQ\qS_R)^{-1}-(\qI_{n_T}+\qS_E\qQ)^{-1}\qS_E\nnb\\
&=(\qI_{n_T}+\qS_E\qQ)^{-1}{\bs\Delta}(\qI_{n_T}+\qQ\qS_R)^{-1}\nnb\\
&=(\qI_{n_T}+\qS_E\qQ)^{-1}\qP(\qI_{n_T}+\qQ\qS_E)^{-1}\label{appTheta}
}
where $\qP={\bs\Delta}(\qI_{n_T}+(\qI_{n_T}+\qQ\qS_E)^{-1}\qQ{\bs\Delta})^{-1}$.
According to Sylvester's law of inertia \cite[p.223]{Horn}, it suffices to prove that
$\qP$ has positive eigenvalue. If ${\bs\Delta}$ is nonsingular, we can write
\equ{
\qP=({\bs\Delta}^{-1}+(\qI_{n_T}+\qQ\qS_E)^{-1}\qQ)^{-1}.
}
Note that $(\qI_{n_T}+\qQ\qS_E)^{-1}\qQ\succeq 0$
and using the assumption ${\bs\Delta}$ is indefinite or positive semi-definite,
hence, we get that $\qP$ has positive eigenvalue. If ${\bs\Delta}$ is singular, there exists $\delta_0> 0$ such that
for any $0<\delta<\delta_0$, ${\bs\Delta}'={\bs\Delta}-\delta\qI_{n_T}$ is nonsingular and also indefinite or positive semi-definite.
Similarly, we can prove that $\qP'={\bs\Delta}'(\qI_{n_T}+(\qI_{n_T}+\qQ\qS_E)^{-1}\qQ{\bs\Delta}')^{-1}$
has positive eigenvalue. Next we prove $\qP\succ \qP'$. Denote $\qW=(\qI_{n_T}+\qQ\qS_E)^{-1}\qQ$.
Similar to the skill in (\ref{appTheta}), we get
\equ{
\qP-\qP'=(\qI_{n_T}+{\bs\Delta}\qW)^{-1}\qW_1(\qI_{n_T}+\qW{\bs\Delta})^{-1}
}
where $\qW_1=\delta(\qI_{n_T}-\delta(\qI_{n_T}+\qW{\bs\Delta})^{-1}\qW)^{-1}$.
We can see that when $\delta$ is sufficient small, $\qW_1$ is positive definite. Thus, $\qP\succ \qP'$ and $\qP$
has positive eigenvalue.

Second, we prove the latter part.
Denote the eigen-decomposition ${\bs\Theta}=\qU_{\Theta}\qD_{\Theta}\qU_{\Theta}^\dagger$
and let $\qQ_1=\qU_{\Theta}^\dagger\qQ\qU_{\Theta}$. We know $\qQ_1\succeq 0$ and $\mathrm{Tr}(\qQ_1)=\mathrm{Tr}(\qQ)$. With these, we can write
\equ{
\mathrm{Tr}(\qQ{\bs\Theta})=\mathrm{Tr}(\qQ_1\qD_{\Theta})=\mathrm{Tr}(\mathrm{diag}(\qQ_1)\qD_{\Theta})
\le \mathrm{Tr}(\qQ_1)\lambda_{\max}(\qD_{\Theta})=\mathrm{Tr}(\qQ)\lambda_{\max}(\bs\Theta).
}
In particular, if $\mathrm{Tr}(\qQ)=1$, we get $\mathrm{Tr}(\qQ{\bs\Theta})\le \lambda_{\max}(\bs\Theta)$.

\section{Proof of Theorem \ref{theo:1}}\label{proofTheo1}

It follows from (\ref{KKTcond2}) that ${\bs \Psi}\qQ=\qQ{\bs\Psi}=0$, that is, ${\bs\Psi}$ and $\qQ$ commute and
have the same eigenvectors \cite[p.239]{Davis} and their eigenvalue patterns are complementary
in the sense that if $\lambda_i(\qQ) > 0$, then $\lambda_i({\bs \Psi}) = 0$, and vice
versa \cite{Vu}.
This result, when combined with (\ref{KKTcond1}), implies that ${\bs\Theta}$ and $\qQ$
commute and have the same eigenvectors, i.e., they have the eigen-decompositions
$\qQ=\qU_q\qD_q\qU_q^\dagger$ and ${\bs \Theta}=\qU_q\qD_{\Theta}\qU_q^\dagger$.
Further, we get ${\bs\Theta}\qQ=\qQ{\bs \Theta}=\theta \qQ$, which, when combined with $\mathrm{Tr}(\qQ)=1$ and
the fact $\mathrm{Tr}(\qQ{\bs \Theta})=\mathrm{Tr}(\qQ^{1/2}{\bs \Theta}\qQ^{1/2})$ is always real, leads to
$\theta=\mathrm{Tr}(\qQ{\bs \Theta})$ and (\ref{basicEq1}) (also see \cite{Li}).

The condition (\ref{basicEq1}) reveals that for the  optimal $\qQ$,  $\qQ{\bs \Theta}$ is a scaled version  of $\qQ$.
Further the eigenvalues of ${\bs\Theta}$ corresponding to the positive eigenvalues of $\qQ$ are all equal to
 $\mathrm{Tr}(\qQ{\bs \Theta})$, while the remaining eigenvalues of ${\bs\Theta}$ are all less than  or equal
to $\mathrm{Tr}(\qQ{\bs \Theta})$, which follows from (\ref{KKTcond1}), (\ref{basicEq1}) and ${\bs\Psi}\succeq 0$.
Based on the above, it holds the second condition (\ref{basicEq2}).

\section{Proof of Theorem \ref{theo:MISOopt}}\label{proofTheoMISOopt}

In the MISO wiretap channel, $\qS_R=\rho\qh_R\qh_R^\dagger$ and $\qS_E=\rho\qh_E\qh_E^\dagger$.
By using the matrix inverse formula for two vectors $\qx$ and $\qy$
\equ{
(\qI+\qx\qy^\dagger)^{-1}=\qI-\qx\qy^\dagger/(1+\qy^\dagger\qx),\label{matrixinverse}
}
We can write
\equ{
{\bs\Theta}=\frac{\rho\qh_R\qh_R^\dagger}{1+\rho\qh_R^\dagger\qQ\qh_R}-
\frac{\rho\qh_E\qh_E^\dagger}{1+\rho\qh_E^\dagger\qQ\qh_E}.
}
That is to say, ${\bs\Theta}$ has the form of $\alpha_1\qh_R\qh_R^\dagger-\alpha_2\qh_E\qh_E^\dagger$,
$\alpha_1>0, \alpha_2>0$. According to Lemma \ref{lem:2}, we know: if $\qh_R=\xi\qh_E$, then ${\bs\Theta}$
has only one nonzero eigenvalue; if $\qh_R\ne \qh_E$, then ${\bs\Theta}$ has only two nonzero eigenvalues,
one is positive and the other is negative. With this, since $\qQ$ satisfies (\ref{basicEq1}),
it is easy to verify $\qQ$ has rank one. Let $\qQ=\qu\qu^\dagger$ and we have
\al{
{\bs\Theta}&=\frac{\rho\qh_R\qh_R^\dagger}{1+\rho\qh_R^\dagger\qu\qu^\dagger\qh_R}-
\frac{\rho\qh_E\qh_E^\dagger}{1+\rho\qh_E^\dagger\qu\qu^\dagger\qh_E}\\
\mathrm{Tr}(\qQ{\bs\Theta})&=\frac{1}{1+\rho\qh_E^\dagger\qu\qu^\dagger\qh_E}-\frac{1}{1+\rho\qh_R^\dagger\qu\qu^\dagger\qh_R}.
}
Let $\omega_1=1+\rho\qh_R^\dagger\qu\qu^\dagger\qh_R$, $\omega_2=1+\rho\qh_E^\dagger\qu\qu^\dagger\qh_E$.
According to Lemma \ref{lem:2}, the largest eigenvalue of ${\bs\Theta}$ is given by
\equ{
\lambda_{\max}({\bs\Theta})=\frac{\rho\|\qh_R\|^2}{2\omega_1}-\frac{\rho\|\qh_E\|^2}{2\omega_2}
+\frac{1}{2}\sqrt{\left(\frac{\rho\|\qh_R\|^2}{\omega_1}+\frac{\rho\|\qh_E\|^2}{\omega_2}\right)^2
-\frac{4\rho^2|\qh_R^\dagger\qh_E|^2}{\omega_1\omega_2}}.
}
Since $\qQ$ satisfies (\ref{basicEq2}), we have
\equ{
\frac{\rho\|\qh_R\|^2}{2\omega_1}-\frac{\rho\|\qh_E\|^2}{2\omega_2}
+\frac{1}{2}\sqrt{\left(\frac{\rho\|\qh_R\|^2}{\omega_1}+\frac{\rho\|\qh_E\|^2}{\omega_2}\right)^2
-\frac{4\rho^2|\qh_R^\dagger\qh_E|^2}{\omega_1\omega_2}}=\frac{1}{\omega_2}-\frac{1}{\omega_1}
}
which leads to
\equ{
1+\frac{\rho\|\qh_R\|^2}{2}-\left(1+\frac{\rho\|\qh_E\|^2}{2}\right)\frac{\omega_1}{\omega_2}
+\frac{\rho}{2}\sqrt{\left(\|\qh_R\|^2+\|\qh_E\|^2\frac{\omega_1}{\omega_2}\right)^2
-4|\qh_R^\dagger\qh_E|^2\frac{\omega_1}{\omega_2}}=0.\label{equomega1omega2}
}
This equation (\ref{equomega1omega2}) is exactly (\ref{equFAlpha}), i.e., $F(\alpha)=0$ where $\alpha=\omega_1/\omega_2$.
On the other hand, we know
\equ{
\frac{\omega_1}{\omega_2}=\frac{1+\rho\qh_R^\dagger\qu\qu^\dagger\qh_R}{1+\rho\qh_E^\dagger\qu\qu^\dagger\qh_E}.\label{omega12omega2}
}
According to the result in \S\ref{sec:3}, the root of $F(\alpha)=0$ corresponds to the maximization of the right hand side (RHS) of
(\ref{omega12omega2}) (see also (\ref{SecRateMaxMISO1})). Thus, the conditions of Theorem \ref{theo:1}
guarantee the optimal input covariance.

\section{Proof of Property \ref{prop:1}}\label{proofProp1}

From (\ref{Theta2}), we know that $\qQ{\bs \Theta}=(\qI_{n_T}+\qQ\qS_E)^{-1}-(\qI_{n_T}+\qQ\qS_R)^{-1}$.
With this, Left-multiplication by $\qI_{n_T}+\qQ\qS_E$ and right-multiplication by $(\qI_{n_T}+\qQ\qS_R)\qQ$
 of both sides of (\ref{basicEq1})
results in
\equ{
\qQ(\qS_R-\qS_E)\qQ=\mathrm{Tr}(\qQ{\bs \Theta})(\qQ+\qQ\qS_E\qQ)(\qQ+\qQ\qS_R\qQ).\label{thetalargerzero1}
}
Note that the left hand side (LHS) of (\ref{thetalargerzero1}) is Hermitian, hence,
$(\qQ+\qQ\qS_E\qQ)(\qQ+\qQ\qS_R\qQ)$ is Hermitian as well, which implies that the matrices
$\qQ+\qQ\qS_E\qQ$ and $\qQ+\qQ\qS_R\qQ$ commute and have the same eigenvectors.
On the other hand, the RHS of (\ref{thetalargerzero1}) has all non-negative eigenvalues,
thus $\qQ(\qS_R-\qS_E)\qQ\succeq 0$.

Finally, recalling (\ref{appTheta}), ${\bs\Theta}=(\qI_{n_T}+\qS_E\qQ)^{-1}(\qS_R-\qS_E)(\qI_{n_T}+\qQ\qS_R)^{-1}$,
we know $\mathrm{rank}(\bs\Theta)=\mathrm{rank}(\qS_R-\qS_E)$ which follows from
the fact $\mathrm{rank}(\qA)+\mathrm{rank}(\qB)-n\le \mathrm{rank}(\qA\qB)\le \min\{\mathrm{rank}(\qA),
\mathrm{rank}(\qB)\}$ for two $n\times n$ matrices $\qA$ and $\qB$ \cite{Horn}.
Further, since $\mathrm{rank}(\qQ{\bs\Theta})\le \min\{\mathrm{rank}(\qQ), \mathrm{rank}(\bs\Theta)\}$,
if $\mathrm{rank}(\bs\Theta)<\mathrm{rank}(\qQ)$, then $\mathrm{rank}(\qQ{\bs\Theta})<\mathrm{rank}(\qQ)$,
but this violates (\ref{basicEq1}). Thus, $\mathrm{rank}(\bs\Theta)\ge\mathrm{rank}(\qQ)$.

\section{Proof of Lemma \ref{lem:rankQ}}\label{proofLemRankQ}

First, we can prove that $\qu_0$ is the eigenvector of $\qS_0$.
To prove this, it suffices to show that $\qS_0\qu_0=\eta_0\qu_0$ for a certain scarlar $\eta_0$.
By using the formula (\ref{matrixinverse}), we can write
\equ{
\qS_0\qu_0=\frac{\qS_1\qu_0}{1+\rho\qu_0^\dagger\qS_1\qu_0}-\frac{\qS_2\qu_0}{1+\rho\qu_0^\dagger\qS_2\qu_0}.\label{S0u0}
}
Since $\qu_0$ is the eigenvector of $(\qI_2+\rho\qS_2)^{-1}(\qI_2+\rho\qS_1)$ associated with the eigenvalue $\lambda_0$, it holds
$(\qI_2+\rho\qS_1)\qu_0=\lambda_0 (\qI_2+\rho\qS_2)\qu_0$ which leads to two facts:
$1+\rho\qu_0^\dagger\qS_1\qu_0=\lambda_0(1+\rho\qu_0^\dagger\qS_2\qu_0)$ and $(\qS_1-\lambda_0\qS_2)\qu_0=\frac{\lambda_0-1}{\rho}\qu_0$.
Inserting them into (\ref{S0u0}) gives $\qS_0\qu_0=\eta_0\qu_0$ where
$\eta_0=(\lambda_0-1)/(\rho(1+\rho\qu_0^\dagger\qS_1\qu_0))$. We also know $\eta_0=\qu_0^\dagger\qS_0\qu_0$.

Second, let $\qv$ is any unit-norm vector and $\qv\ne \xi\qu_0$.
Define the function
\al{
g(t, \qv)&=\log\det(\qI_2+\rho((1-t)\qu_0\qu_0^\dagger+t\qv\qv^\dagger)\qS_1)\nnb\\
&\quad-\log\det(\qI_2+\rho((1-t)\qu_0\qu_0^\dagger+t\qv\qv^\dagger)\qS_2),\ t\in [0,1].
}
It is easy to know: $(1-t)\qu_0\qu_0^\dagger+t\qv\qv^\dagger \in \Omega$;
when $t=0$, $(1-t)\qu_0\qu_0^\dagger+t\qv\qv^\dagger=\qu_0\qu_0^\dagger$;
when $t>0$, $(1-t)\qu_0\qu_0^\dagger+t\qv\qv^\dagger$ has rank two.
We can state that if $\frac{\partial g}{\partial t}|_{t=0}> 0$ for a certain $\qv_0$, the optimal $\qQ$ has rank two.
The reason is simple: assume that the optimal $\qQ$ has rank one, then it holds $g(0,\qv)\ge g(t,\qv)$, $\forall t, \qv$.
But on the other hand, it follows from $\frac{\partial g}{\partial t}|_{t=0}> 0$
that there exists $t_0>0$ such that $g(t_0, \qv_0)>g(0,\qv_0)$.
This produces a contradiction. By using the derivative formula \cite{Jafar}
\equ{
\frac{\partial \log\mathrm{det}(\qA+t\qB)}{\partial t}=\mathrm{Tr}\{\qB(\qA+t\qB)^{-1}\},
}
we can obtain
\equ{
\frac{\partial g}{\partial t}\big\vert_{t=0}=\mathrm{Tr}\{\rho(\qv\qv^\dagger-\qu_0\qu_0^\dagger)\qS_0\}
=\rho(\qv^\dagger\qS_0\qv-\qu_0^\dagger\qS_0\qu_0).\label{diffg}
}
It follows from (\ref{diffg}) that if $\qu_0^\dagger\qS_0\qu_0<\lambda_{\max}(\qS_0)$, there always exists $\qv$
such that $\frac{\partial g}{\partial t}|_{t=0}> 0$, and hence the optimal $\qQ$ has rank two;
if the optimal $\qQ$ has rank one, there must be
$\frac{\partial g}{\partial t}|_{t=0}\le 0$ which leads to $\qu_0^\dagger\qS_0\qu_0=\lambda_{\max}(\qS_0)$.
Further, since $\qu_0$ is the eigenvector of $\qS_0$, thus
$\qu_0^\dagger\qS_0\qu_0<\lambda_{\max}(\qS_0)$ and $\qu_0^\dagger\qS_0\qu_0=\lambda_{\max}(\qS_0)$
are equivalent to $2\qu_0^\dagger\qS_0\qu_0<\mathrm{Tr}(\qS_0)$ and $2\qu_0^\dagger\qS_0\qu_0\ge \mathrm{Tr}(\qS_0)$,
respectively. This completes the proof.

\section{Proof of Theorem \ref{theo:rankQ}}\label{proofTheoRankQ}

We discuss two cases respectively.
\begin{enumerate}
    \item [(i)] $\qS_2$ has full rank. Since $(\qI_2+\rho\qS_2)^{-1}(\qI_2+\rho\qS_1)=(\qI_2/\rho+\qS_2)^{-1}(\qI_2/\rho+\qS_1)$,
we know that as $\rho\to\infty$, $\qu_0$ goes to the eigenvector $\wt{\qu}_0$ associated with
the largest eigenvalue $\lambda$ of $\qS_2^{-1}\qS_1$.
We can write $\qS_1\wt{\qu}_0=\lambda\qS_2\wt{\qu}_0$,
hence $\wt{\qu}_0^\dagger\qS_1^2\wt{\qu}_0=\lambda\wt{\qu}_0^\dagger\qS_1\qS_2\wt{\qu}_0$,
$\wt{\qu}_0^\dagger\qS_1\wt{\qu}_0=\lambda\wt{\qu}_0^\dagger\qS_2\wt{\qu}_0$. Then
\al{
g(\rho)&\to \frac{\wt{\qu}_0^\dagger\qS_1^2\wt{\qu}_0}{\wt{\qu}_0\qS_1\wt{\qu}_0}
-\frac{\wt{\qu}_0\qS_2^2\wt{\qu}_0}{\wt{\qu}_0\qS_2\wt{\qu}_0}-\mathrm{Tr}(\qS_1-\qS_2)
=\frac{\wt{\qu}_0^\dagger(\qS_1-\qS_2)\qS_2\wt{\qu}_0}{\wt{\qu}_0^\dagger\qS_2\wt{\qu}_0}
-\mathrm{Tr}(\qS_1-\qS_2)\nnb\\
&\le \lambda_{\max}(\qS_1-\qS_2)-\mathrm{Tr}(\qS_1-\qS_2)< 0
}
where we use the fact $\qx^\dagger\qA\qx/(\qx^\dagger\qB\qx)\le \lambda_{\max}(\qA\qB^{-1})$.
    \item [(ii)] $\qS_2$ has rank one and hence can be expressed as $\qS_2=\qv_2\qv_2^\dagger$.
As $\rho\to \infty$,
\al{
g(\rho)&\to \frac{\qu_0^\dagger\qS_1^2\qu_0}{\qu_0^\dagger\qS_1\qu_0}
-\frac{\qu_0^\dagger\qS_2^2\qu_0}{\qu_0^\dagger\qS_2\qu_0}-\mathrm{Tr}(\qS_1-\qS_2)
=\frac{\qu_0^\dagger\qS_1^2\qu_0}{\qu_0^\dagger\qS_1\qu_0}
-\mathrm{Tr}(\qS_2)-\mathrm{Tr}(\qS_1-\qS_2)\nnb\\
&\le \lambda_{\max}(\qS_1)-\mathrm{Tr}(\qS_1)<0.
}
\end{enumerate}

\section{Proof of Theorem \ref{theo:2}}\label{proofTheo2}

According to Property \ref{prop:1}, we know $\qQ(\qS_R-\qS_E)\qQ\succeq 0$.
We will show that under the condition of Theorem \ref{theo:2},
any $\qQ\succeq 0$ that satisfies $\qQ(\qS_R-\qS_E)\qQ\succeq 0$ has rank one.
To prove this, we use the mathematical induction which consists of two steps:
(1) showing that the statement holds when $n_T=2$;
(2) showing that if the statement holds for some $n_T=k\ge 2$, then the statement also holds when
$k+1$ is substituted for $k$.

First, we address (1). It follows from Property \ref{prop:2} that the optimal $\qQ$ has low rank, i.e., rank one.
Next we deal with (2).
We denote the eigen-decomposition $\qS_R-\qS_E=\qU\qD\qU^\dagger$ where $\qD$
is a diagonal matrix with all negative diagonal entries except a positive one.
Let $\qX=\qU^\dagger\qQ\qU$. Then $\qQ(\qS_R-\qS_E)\qQ\succeq 0$ is equivalent to $\qX\qD\qX\succeq 0$.
It suffices to prove the following:

\begin{quote}
{\it Problem}: Assume the following is true: any $k\times k$ diagonal matrix $\qD$ which has $k-1$ negative
entries and one positive one in its diagonal, then any $k\times k$ matrix $\qX\succeq 0$ that satisfies
$\qX\qD\qX\succeq 0$ always has rank one. Is it true for $k+1$?
\end{quote}
Denote the $(k+1)\times (k+1)$ diagonal matrix
\equ{
\qD=\left(%
\begin{array}{cc}
  \qD_1 & \mb{0} \\
  \mb{0}^\dagger & -\xi \\
\end{array}%
\right)
}
where $\qD_1$ is a $k\times k$ diagonal matrix with all negative diagonal entries except a positive one, $\xi>0$.
We also denote the $(k+1)\times (k+1)$ positive semi-definite matrix
\equ{
\qX=\left(%
\begin{array}{cc}
  \qX_1 & \mb{b} \\
  \mb{b}^\dagger & x \\
\end{array}%
\right)\label{X}
}
where $\qX_1$ is a $k\times k$ matrix, $\mb{b}$ is a $k\times 1$ vector, $x$ is a scalar.
It follows from $\qX\succeq 0$ that: $\qX_1\succeq 0$, $x\ge 0$; if $x=0$, then $\mb{b}=\mb{0}$ \cite{Horn}.
Now we can write
\equ{
\qX\qD\qX=\left(%
\begin{array}{cc}
  \qX_1\qD_1\qX_1-\xi \mb{b}\mb{b}^\dagger & \qX_1\qD_1\mb{b}-\xi x\mb{b} \\
   \mb{b}^\dagger\qD_1\qX1-\xi x\mb{b}^\dagger & \mb{b}^\dagger\qD_1\mb{b}-\xi x^2 \\
\end{array}%
\right).\label{XDX0}
}
When $x=0$, there will be $\mb{b}=0$ and hence $\qX\qD\qX=\mathrm{diag}(\qX_1\qD_1\qX_1, 0)\succeq 0$ which is equivalent to
$\qX_1\qD_1\qX_1\succeq 0$. With this, based on the assumption for $k$,
 $\qX_1$ has rank one, thus, $\qX$ has rank one. In the following, we consider $x>0$.
Combining (\ref{XDX0}) with $\qX\qD\qX\succeq 0$ results in $\qX_1\qD_1\qX_1-\xi \mb{b}\mb{b}^\dagger\succeq 0$
and $\mb{b}^\dagger\qD_1\mb{b}-\xi x^2\ge 0$ \cite{Horn}. From $\mb{b}^\dagger\qD_1\mb{b}-\xi x^2\ge 0$,
we get $\mb{b}^\dagger\qD_1\mb{b}\ge\xi x^2> 0$ (hence $\mb{b}\ne 0$).
It follows from $\qX_1\qD_1\qX_1-\xi \mb{b}\mb{b}^\dagger\succeq 0$ that
$\qX_1\qD_1\qX_1 \succeq 0$. With this, by using the assumption for $k$,
we know $\qX_1$ has rank one and can be expressed as $\qX_1=\qv\qv^\dagger$.
Thus, $\qX_1\qD_1\qX_1=(\qv^\dagger\qD_1\qv)\qv\qv^\dagger$ and $\qv^\dagger\qD_1\qv\ge 0$. We can prove that $\qv=c_1\mb{b}$ for a certain scalar $c_1$.
To see why this is the case, let us rewrite
$\qX_1\qD_1\qX_1-\xi \mb{b}\mb{b}^\dagger=(\qv^\dagger\qD_1\qv)\qv\qv^\dagger-\xi \mb{b}\mb{b}^\dagger$.
Note that $\qv^\dagger\qD_1\qv\ge 0$, hence according to Lemma \ref{lem:2}, if $\qv\ne c_1\mb{b}$,
then $(\qv^\dagger\qD_1\qv)\qv\qv^\dagger-\xi \mb{b}\mb{b}^\dagger$ always has a negative eigenvalue.
But this violates $\qX_1\qD_1\qX_1-\xi \mb{b}\mb{b}^\dagger\succeq 0$. Thus, $\qv=c_1\mb{b}$.
Next, we determine $c_1$. Now we can write $\qX_1\qD_1\qX_1-\xi \mb{b}\mb{b}^\dagger=(\mb{b}^\dagger\qD_1\mb{b}|c_1|^4-\xi)\mb{b}\mb{b}^\dagger$
and
\equ{
\qX\qD\qX=\left(%
\begin{array}{cc}
  (\beta|c_1|^4-\xi)\mb{b}\mb{b}^\dagger & (\beta|c_1|^2-\xi x)\mb{b} \\
   (\beta|c_1|^2-\xi x)\mb{b}^\dagger & \beta-\xi x^2 \\
\end{array}%
\right)\label{XDX}
}
where $\beta=\mb{b}^\dagger\qD_1\mb{b}$.
Since $\qX\qD\qX\succeq 0$, we have
\equ{
\beta|c_1|^4-\xi\ge 0, (\beta|c_1|^4-\xi)(\beta-\xi x^2)\ge (\beta|c_1|^2-\xi x)^2
}
which results in $|c_1|^2=1/x$.
With this, we can rewrite (\ref{X}) as
\equ{
\qX=\left(%
\begin{array}{cc}
  |c_1|^2\mb{b}\mb{b}^\dagger & \mb{b} \\
  \mb{b}^\dagger & x \\
\end{array}%
\right)=\frac{1}{x}\left(%
\begin{array}{c}
  \mb{b} \\
  x \\
\end{array}%
\right)\left(%
\begin{array}{c}
  \mb{b} \\
  x \\
\end{array}%
\right)^\dagger
}
which leads to $\qX$ has rank one. This completes the proof.

\begin{figure}[h]
\centering
\includegraphics[width=3.5in]{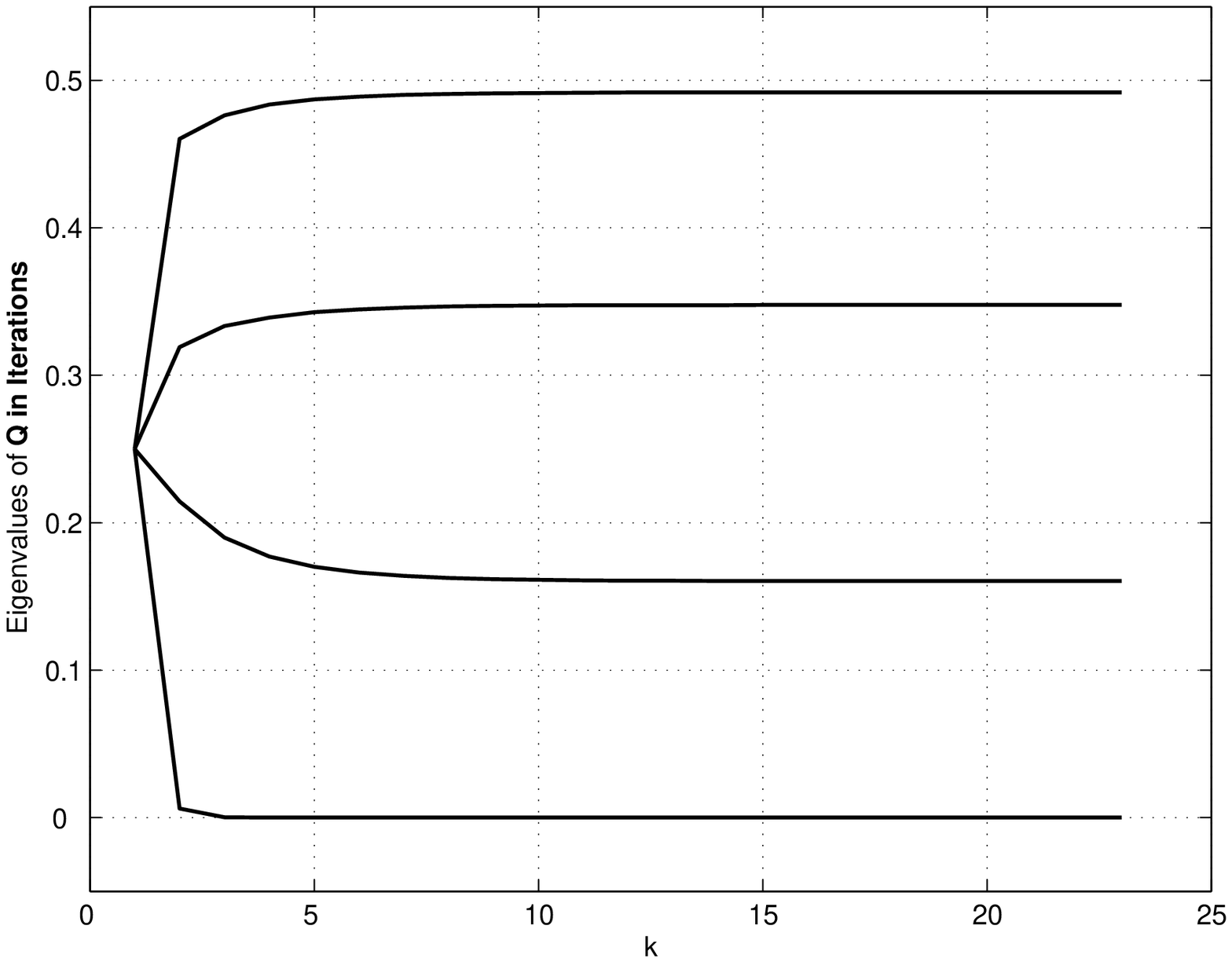}
\caption{Eigenvalues of $\qQ^k$ in iterations, positive definite case.}
\label{fig:1}
\end{figure}

\begin{figure}[h]
\centering
\includegraphics[width=3.5in]{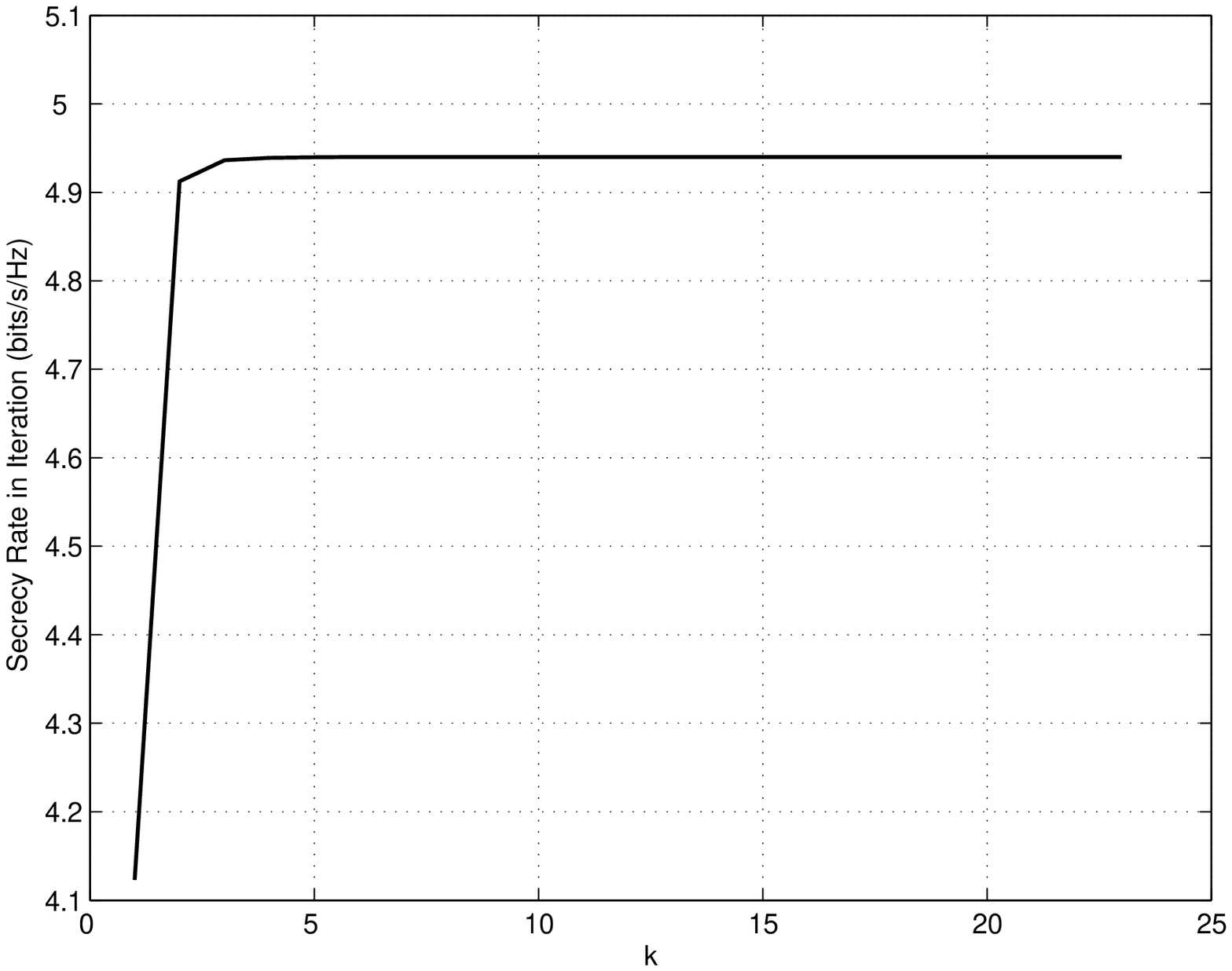}
\caption{Secrecy rate in iterations, positive definite case.}
\label{fig:2}
\end{figure}

\begin{figure}[h]
\centering
\includegraphics[width=3.5in]{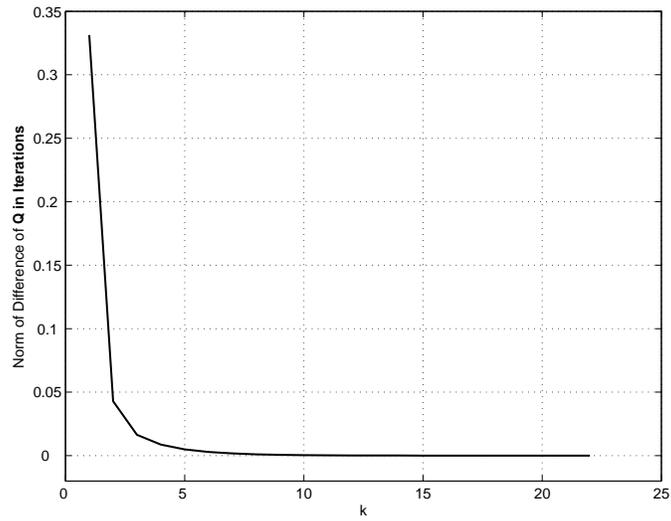}
\caption{$\|\qQ^{k+1}-\qQ^k\|$, positive definite case.}
\label{fig:3}
\end{figure}

\begin{figure}[h]
\centering
\includegraphics[width=3.5in]{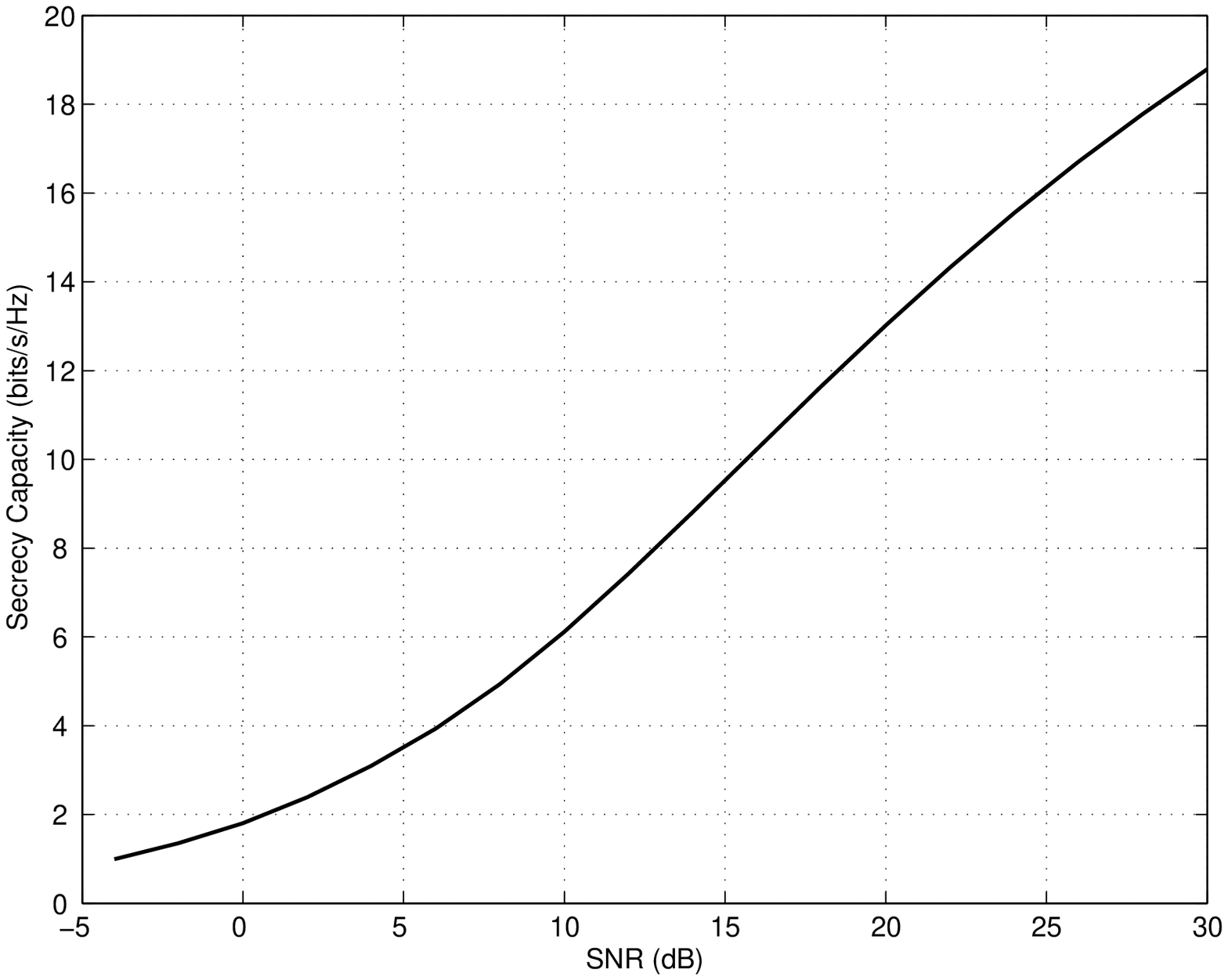}
\caption{Secrecy capacity Vs. SNR, positive definite case.}
\label{fig:4}
\end{figure}

\begin{figure}[h]
\centering
\includegraphics[width=3.5in]{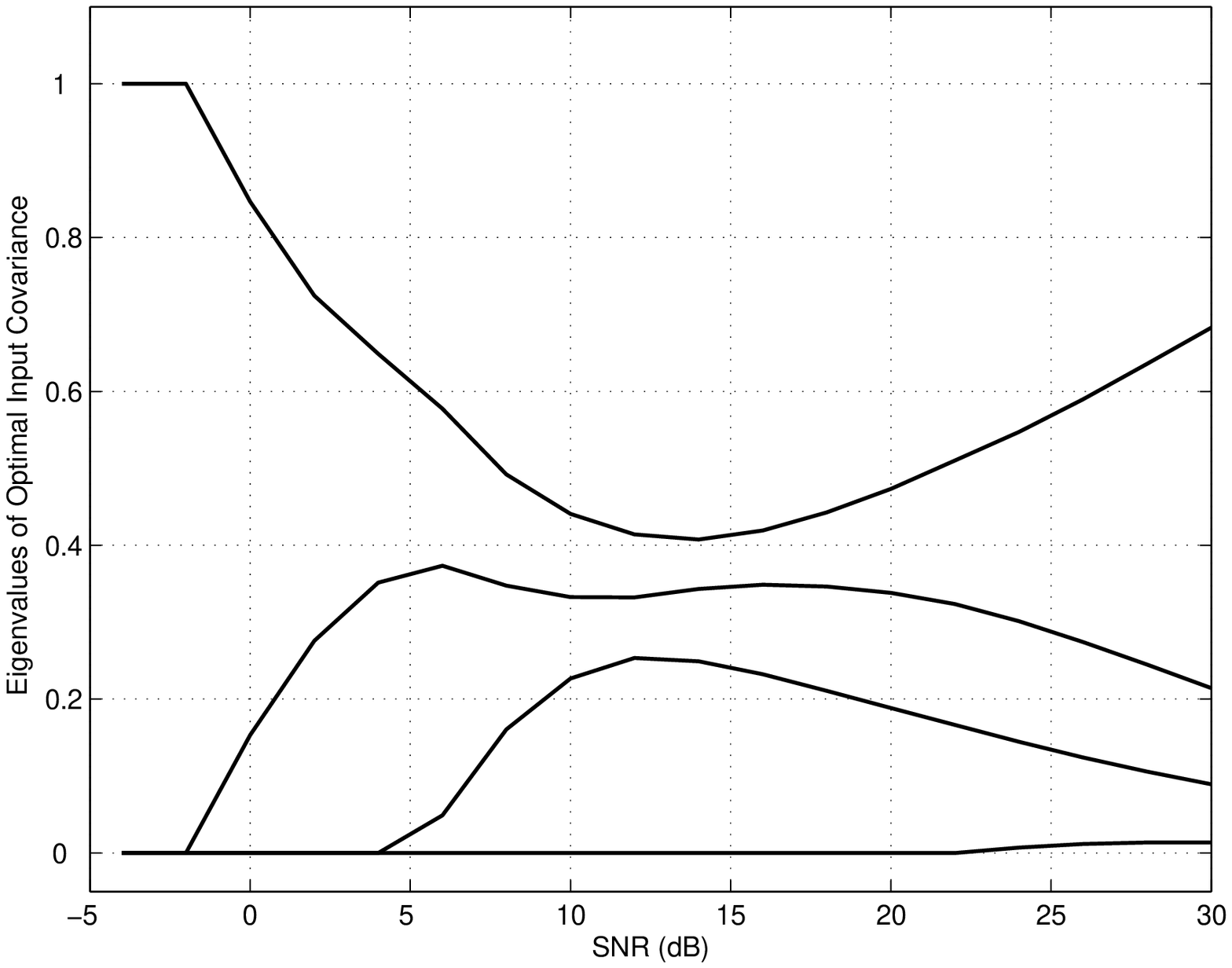}
\caption{Eigenvalues of optimal input covariance Vs. SNR, positive definite case.}
\label{fig:5}
\end{figure}

\begin{figure}[h]
\centering
\includegraphics[width=3.5in]{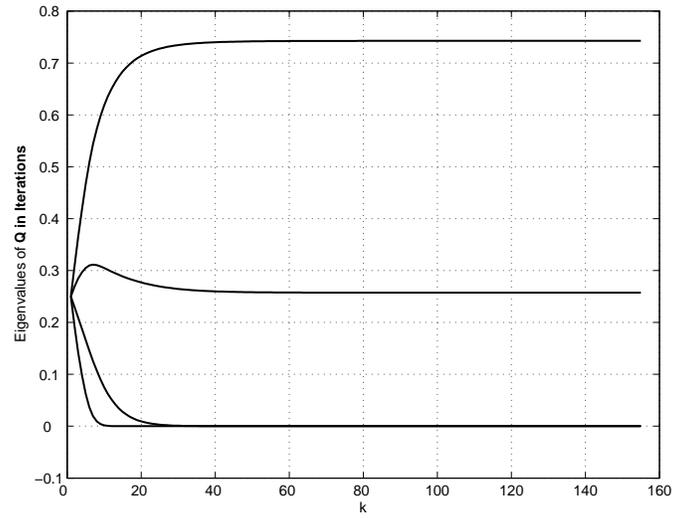}
\caption{Eigenvalues of $\qQ^k$ in iterations, non positive definite case.}
\label{fig:6}
\end{figure}

\begin{figure}[h]
\centering
\includegraphics[width=3.5in]{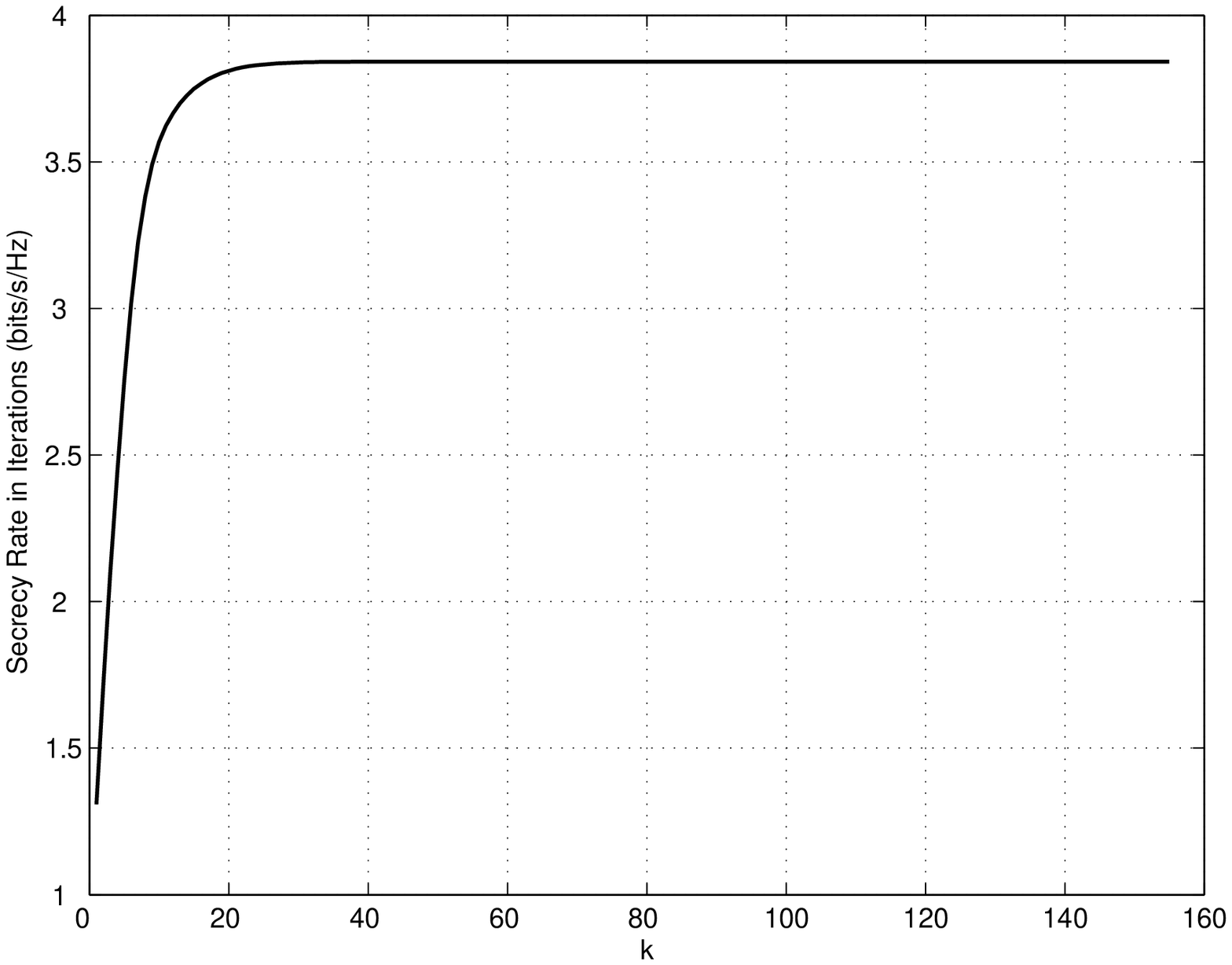}
\caption{Secrecy rate in iterations, non positive definite case.}
\label{fig:7}
\end{figure}

\begin{figure}[h]
\centering
\includegraphics[width=3.5in]{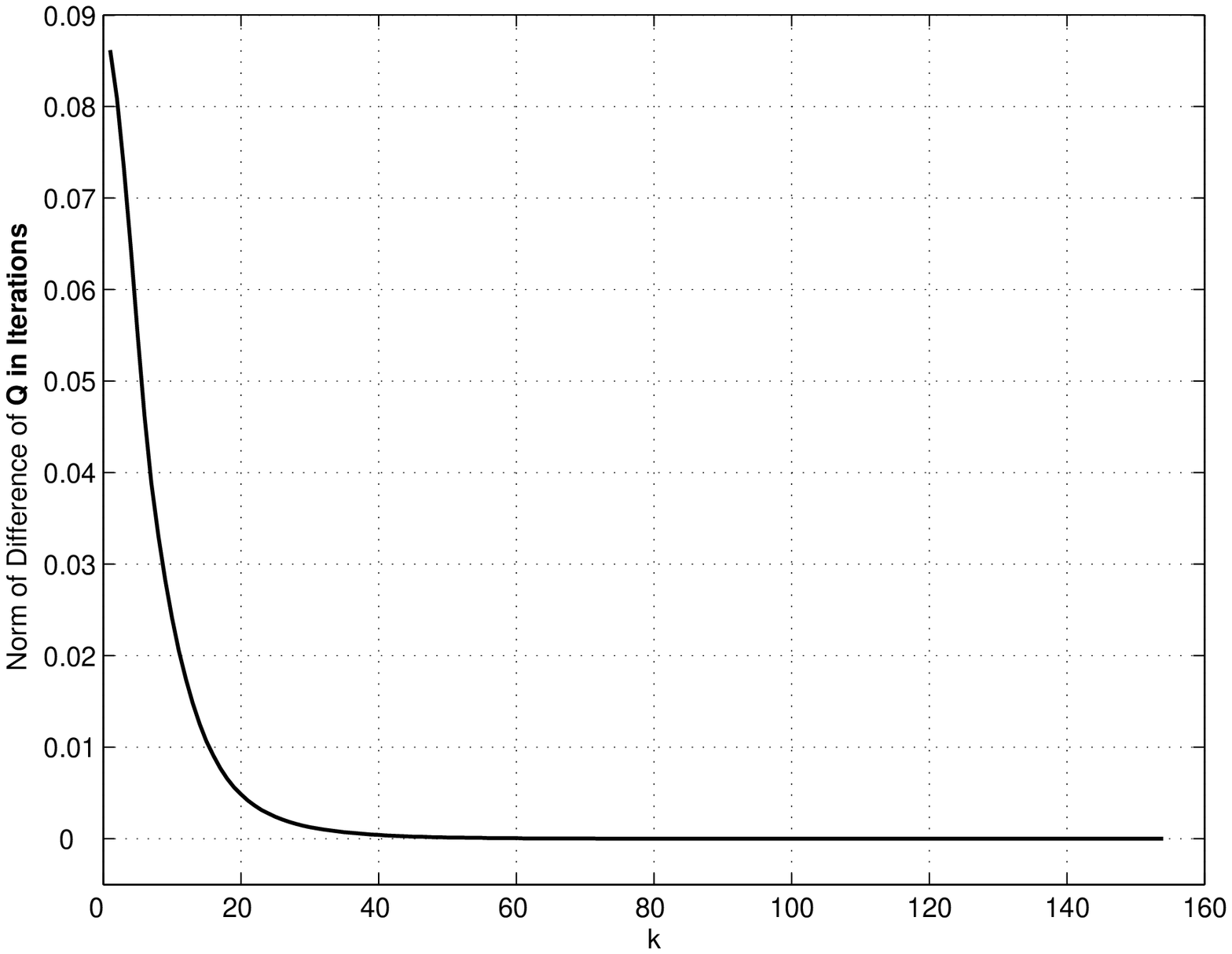}
\caption{$\|\qQ^{k+1}-\qQ^k\|$, non positive definite case.}
\label{fig:8}
\end{figure}

\begin{figure}[h]
\centering
\includegraphics[width=3.5in]{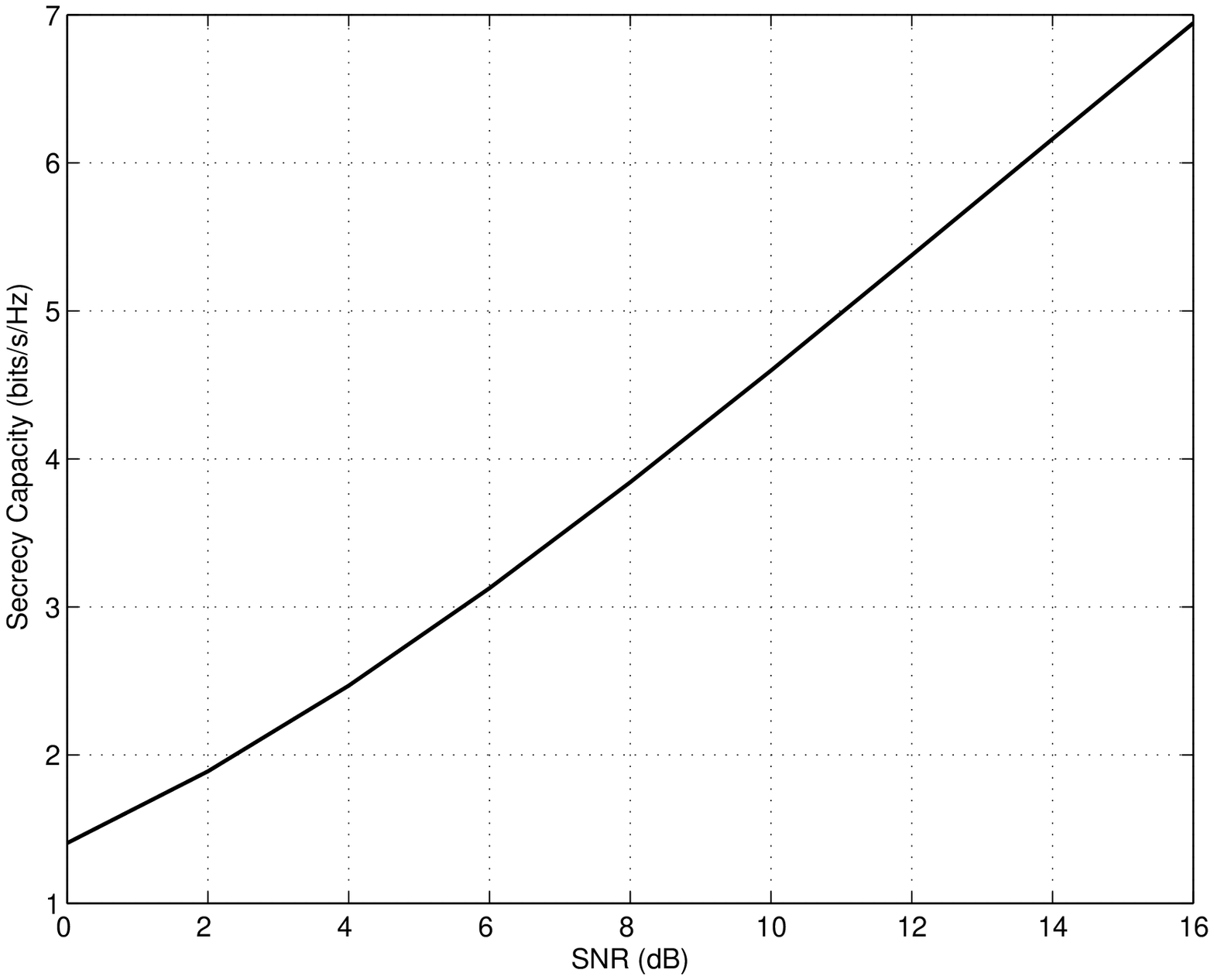}
\caption{Secrecy capacity Vs. SNR, non positive definite case.}
\label{fig:9}
\end{figure}

\begin{figure}[h]
\centering
\includegraphics[width=3.5in]{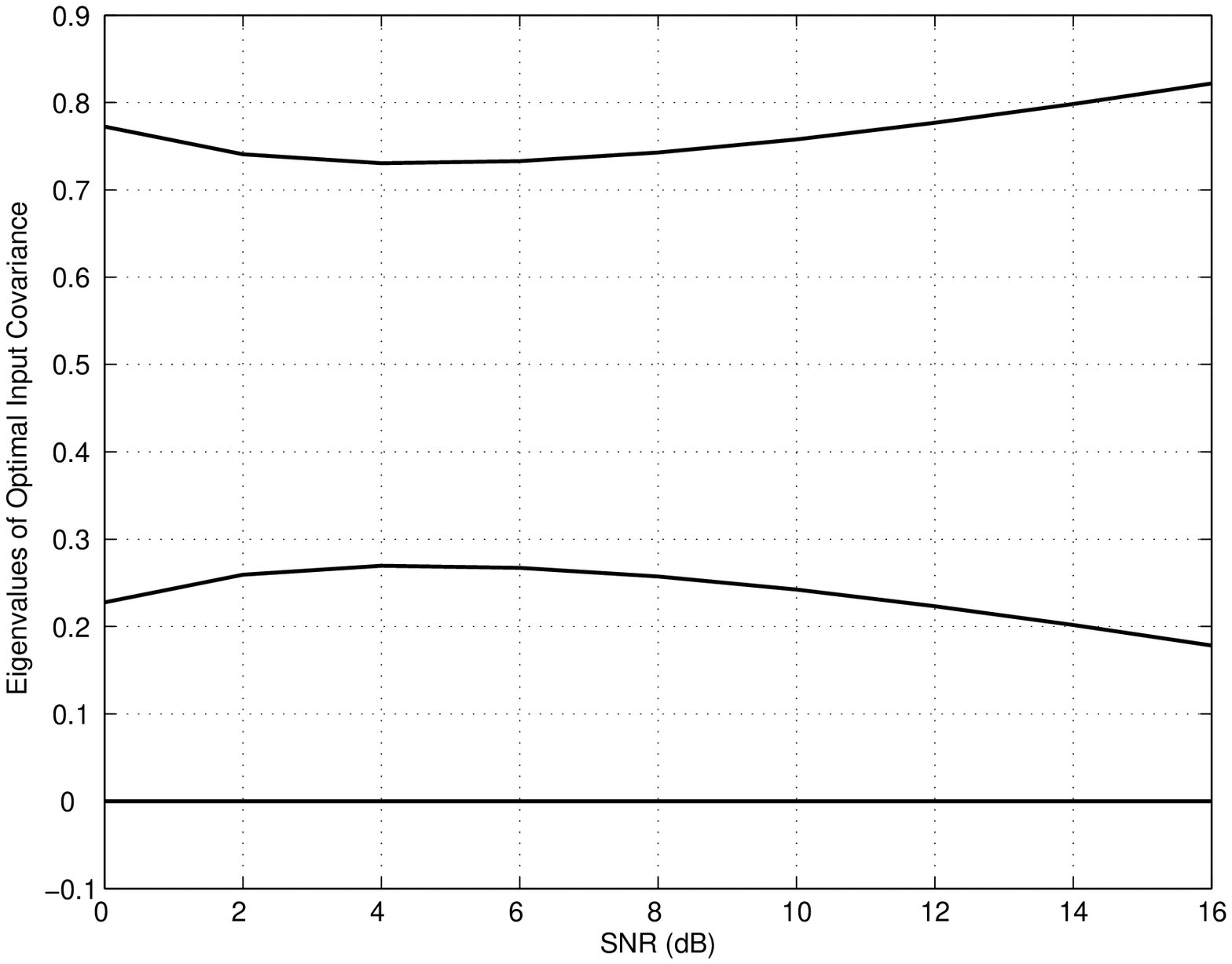}
\caption{Eigenvalues of optimal input covariance Vs. SNR, non positive definite case.}
\label{fig:10}
\end{figure}

\end{document}